\documentclass[a4paper,12pt]{article}
\usepackage[margin=1.1in]{geometry}
\usepackage[backend=biber,doi=false,url=false,style=apa,sortcites=false]{biblatex}
\usepackage[titletoc, title]{appendix}
\usepackage[colorlinks,citecolor=blue,urlcolor=blue,linkcolor=blue]{hyperref}
\usepackage{tikz}
\usepackage{amssymb} 
\usepackage{kbordermatrix} 
\usepackage{amsmath} 
\usepackage{mathtools} 
\usepackage{amsthm} 
\usepackage{gensymb} 
\usepackage[skip=10pt,font=small,labelfont=bf]{caption} 
\usepackage{subcaption} 
\usepackage{enumitem} 
\usepackage{pdfpages}

\DeclareMathOperator*{\argmax}{arg\,max}

\DeclareMathOperator{\I}{I}
\DeclareMathOperator{\II}{II}
\DeclareMathOperator{\III}{III}

\theoremstyle{definition}
\newtheorem{axiom}{Axiom}
\newtheorem{proposition}{Proposition}
\newtheorem{definition}{Definition}
\newtheorem{lemma}{Lemma}
\newtheorem{theorem}{Theorem}

\newtheorem{axiomalt}{Axiom}[axiom]
\newenvironment{axiomp}[1]{
  \renewcommand\theaxiomalt{#1}
  \axiomalt
}{\endaxiomalt}

\renewcommand{\kbldelim}{(}
\renewcommand{\kbrdelim}{)}

\makeatletter
\def\@fnsymbol#1{\ensuremath{\ifcase#1\or *\or\ddagger\or
   \mathsection\or \mathparagraph\or \|\or **\or \dagger \else\@ctrerr\fi}}
\makeatother

\title{Putting Context into Preference Aggregation\thanks{We thank Hervé Moulin and Yves Sprumont for helpful comments. We gratefully acknowledge financial support from the \textit{Jubiläumsfonds} of the Austrian National Bank (OeNB), project number 18719.}}
\author{Philipp Peitler\thanks{University of Vienna; philipp.peitler@univie.ac.at} \and Karl H. Schlag\thanks{University of Vienna; karl.schlag@univie.ac.at}}
\date{April 30, 2024 \\ \href{https://philipppeitler.github.io/mywebsite/context.pdf}{[Latest version here]}}
\begin{document}
\maketitle

\begin{abstract}
The axioms underlying Arrow's impossibility theorem are very restrictive in terms of what can be used when aggregating preferences. Social preferences may not depend on the menu nor on preferences over alternatives outside the menu. But context matters. So, we weaken these restrictions to allow for context to be included. The context, as we define, describes which alternatives in the menu and which preferences over alternatives outside the menu matter. We obtain unique representations. These are discussed in examples involving markets, the intertemporal well-being of an individual, and bargaining.
\end{abstract}

\newpage
\section{Introduction}

A central question in economics is how a social planner should compare alternatives when the preferences of the individuals involved are in conflict. To determine what is best for those individuals involved, the planner needs a rule that aggregates individual preferences into a single social preference. Different strands of the economic literature use very different rules. In the analysis of markets, consumer surplus is used to evaluate the joint welfare of consumers. In the analysis of dynamically inconsistent decision makers, intertemporal welfare is evaluated either by the Pareto criterion or by the long-run utility \autocite{odonoghue1999doing}. In bargaining, the Nash bargaining solution \autocite{nash1950bargaining} and the Kalai-Smorodinsky bargaining solution \autocite{kalai1975other} are prominent. Note here that bargaining problems can be viewed as the search for the best allocation after aggregating preferences. Even within their respective settings, these criteria can only be applied under narrow modeling assumptions. More importantly, each of them fails to conform to a minimal set of desiderata. Consumer surplus violates the neutrality axiom, as it is sensitive to the labels of social alternatives. The Pareto criterion does not generate complete preferences, and long-run utility violates the Pareto principle. Neither the Nash nor the Kalai-Smorodinsky bargaining solution can be interpreted as resulting from the most preferred allocation under a social preference relation that satisfies the \textcite{vonneumann1944theory} axioms. Some of the above violations are not obvious and will be demonstrated. In each of these settings, we could search separately for other rules. However, we feel that the well-being of individuals should be traded off based on principles that are appealing independently of the application. Therefore, we would like to have a universal rule that can be applied to all settings. In light of \textcite{arrow1950difficulty,arrow1963social}, it is not clear whether such a rule exists. In the following, we argue that some of Arrow's demands are too stringent and should be relaxed. Before we state our contentions, let us briefly revisit Arrow's theorem.

\textcite{arrow1963social} shows that there is no aggregation rule that satisfies completeness, transitivity, the Pareto principle, non-dictatorship, and \textit{independence of irrelevant alternatives} (IIA). If we want the social planner to be rational, benevolent, and impartial, then the first four axioms of Arrow serve as minimal requirements. Consider now the fifth and last axiom, IIA. According to Arrow, IIA demands that when the planner chooses from a menu of alternatives denoted by $S$, the choice $C(S)$ is not allowed to depend on individual preferences over alternatives outside $S$. In addition to the above, Arrow assumes that social choice from any menu is made according to a single social preference over all alternatives. Hence, Arrow implicitly assumes one more axiom: \textit{menu independence} (MI). Together, these axioms make it impossible to aggregate preferences. Weakening either IIA or MI is a potential avenue to escape the impossibility. We now argue that one can weaken either of them as neither is as desirable as the first four axioms mentioned above.

First, we assess IIA with the aid of the following example by \textcite{pearce2021individual}. Imagine being tasked with choosing between two social alternatives, $x$ and $y$, for a group of five kindergarten children. Strict preferences of the children are given in Figure \ref{ch1:fig:ordinal}. \begin{figure}[h]
\center
\begin{tabular}{ c c c c c }
 1 & 2 & 3 & 4 & 5\\ 
 x & x & x & x & y\\  
 y & y & y & y & x\\   
\end{tabular}
\caption{Children's ordinal preferences.}
\label{ch1:fig:ordinal}
\end{figure}
In isolation, one would be inclined to choose $x$ over $y$ as it is preferred by four of the five children. However, we also learn that under $x$, each of the first four children gets 1001 toys, while under $y$, each of them only gets 1000 toys. Moreover, we are told that the fifth child has a fatal illness that would be cured under $y$, whereas under $x$, the fifth child dies a long and terrifying death. Unquestionably, this additional information would change our preference to $y$ over $x$. Pearce, therefore, concludes that information besides individual preferences between $x$ and $y$ must be relevant for social choice. But which information exactly? Let us embed the comparison between $x$ and $y$ into a social choice problem in which, for each of the first four children, there exists an alternative $z_i$ where only this child $i$ dies and every other child gets 1001 toys. While the alternatives $z_1$ to $z_4$ aren't currently feasible (i.e., they are not in the menu), they are still possible (in the sense of conceivable) and could be in the menu in a different situation. Assume that Figure \ref{ch1:fig:cardinal} depicts the von Neumann-Morgenstern preferences of the children over all these alternatives. \begin{figure}[t]
\center
\begin{tikzpicture}
\draw (0,0)--(0,4);
\node[above] at (0,4) {1};
\node[circle,draw,inner sep=1.2,fill=black] at (0,0.3) {};
\node[left] at (0,0.3) {$z_1$};
\node[circle,draw,inner sep=1.2,fill=black] at (0,3.7) {};
\node[left] at (0,3.7) {$z_{i\neq 1}$,$x$};
\node[circle,draw,inner sep=1.2,fill=black] at (0,3.5) {};
\node[right] at (0,3.5) {$y$};

\draw (2,0)--(2,4);
\node[above] at (2,4) {2};
\node[circle,draw,inner sep=1.2,fill=black] at (2,0.3) {};
\node[left] at (2,0.3) {$z_2$};
\node[circle,draw,inner sep=1.2,fill=black] at (2,3.7) {};
\node[left] at (2,3.7) {$z_{i\neq 2}$,$x$};
\node[circle,draw,inner sep=1.2,fill=black] at (2,3.5) {};
\node[right] at (2,3.5) {$y$};

\draw (4,0)--(4,4);
\node[above] at (4,4) {3};
\node[circle,draw,inner sep=1.2,fill=black] at (4,0.3) {};
\node[left] at (4,0.3) {$z_3$};
\node[circle,draw,inner sep=1.2,fill=black] at (4,3.7) {};
\node[left] at (4,3.7) {$z_{i\neq 3}$,$x$};
\node[circle,draw,inner sep=1.2,fill=black] at (4,3.5) {};
\node[right] at (4,3.5) {$y$};

\draw (6,0)--(6,4);
\node[above] at (6,4) {4};
\node[circle,draw,inner sep=1.2,fill=black] at (6,0.3) {};
\node[left] at (6,0.3) {$z_4$};
\node[circle,draw,inner sep=1.2,fill=black] at (6,3.7) {};
\node[left] at (6,3.7) {$z_{i\neq 4}$,$x$};
\node[circle,draw,inner sep=1.2,fill=black] at (6,3.5) {};
\node[right] at (6,3.5) {$y$};

\draw (8,0)--(8,4);
\node[above] at (8,4) {5};
\node[circle,draw,inner sep=1.2,fill=black] at (8,0.3) {};
\node[left] at (8,0.3) {$x$};
\node[circle,draw,inner sep=1.2,fill=black] at (8,3.7) {};
\node[left] at (8,3.7) {$z_i$};
\node[circle,draw,inner sep=1.2,fill=black] at (8,3.5) {};
\node[right] at (8,3.5) {$y$};
\end{tikzpicture}
\caption{Children's cardinal utility.}
\label{ch1:fig:cardinal}
\end{figure}
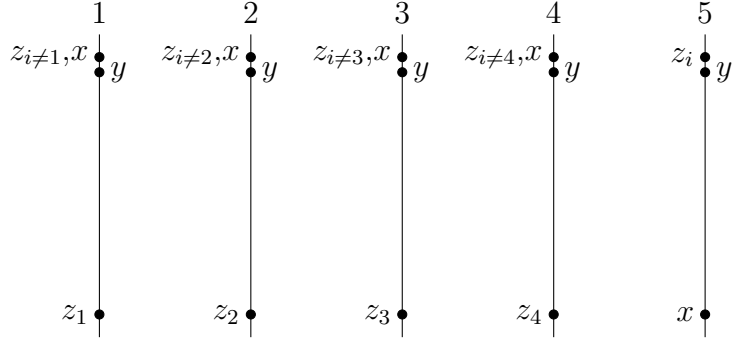 The inclusion of alternatives $z_1$ to $z_4$ provides a context that puts the alternatives $x$ and $y$ into perspective. Adding this context, it becomes clear that going from $y$ to $x$ represents a marginal improvement for the first four children at the expense of a substantial loss to the fifth child. IIA demands that we ignore this additional information, which seems unreasonable. In order to include the context provided by the infeasible alternatives, IIA has to be weakened.

Next, we assess MI. To do this, consider the following example by \textcite{sen1993internal}. An individual faces the menu $\{x,y\}$ where $y$ means taking the last remaining apple from the fruit basket at the dinner table, and $x$ means taking nothing instead. Compare this to the situation where there is a second apple in the basket, such that the menu is $\{x,y,z\}$ where $z$ means taking the other apple. One can plausibly prefer $x$ over $y$ under the first menu and $y$ over $x$ under the second. Similarly, our understanding of fairness in a social choice setting might depend on what is currently feasible. When comparing two social alternatives, $x$ and $y$, information about the feasibility of other alternatives provides a context that helps with the evaluation. In order to make use of that information, MI has to be weakened.

Above we argued that neither IIA nor MI is desirable, as they ignore additional information that can help in the assessment of social alternatives. Weakening either of these conditions opens the possibility for preference aggregation. We are looking for aggregation rules that abide by the von Neumann-Morgenstern axioms, satisfy the strong Pareto condition, and are anonymous. These axioms strengthen Arrow's minimal requirements for rationality, benevolence, and impartiality. We do not want to drop IIA and MI completely, as we wish to limit the information that can be used by the social planner. We, therefore, define two weaker axioms that identify when two social choice problems provide the same relevant information and thus yield the same social preferences. These axioms implicitly define the context as the context captures what is relevant. \textit{Independence of irrelevant comparable alternatives} (\ref{ch1:IICA}) weakens IIA by requiring that preferences over alternatives outside the menu do not influence the ranking if they are comparable for each individual to other alternatives. Specifically, an alternative is comparable to others if each individual is indifferent to some mixture of these other alternatives. \textit{Menu independence of comparable alternatives} (\ref{ch1:MICA}) weakens MI by allowing alternatives to be dropped from the menu if they are comparable to other alternatives in the menu. Each of these axioms, together with our other axioms, uniquely defines an aggregation rule that can be applied to any social choice or aggregation problem. Both of our representations are relative utilitarian, meaning the social welfare of an alternative is identified as the sum of individual utilities. When we weaken MI and maintain IIA, individual utilities are normalized relative to the menu. When we weaken IIA and maintain MI, individual utilities are normalized relative to the set of all possible alternatives within a given setting. Normalizing means setting the utility of the worst alternative in the respective set to 0 and the utility of the best alternative to 1. In a later section, we then use our representations to quantify consumer welfare and the well-being of a dynamically inconsistent decision maker and to derive a fair bargaining solution.

\subsection{Literature}

In the tradition of \textcite{arrow1950difficulty,arrow1963social}, we axiomatize relative utilitarianism in a \textit{multi-profile} setting, meaning that at least one axiom relates social preferences across different profiles of individual preference (e.g., \ref{ch1:IIA} or \ref{ch1:IICA}). Previous characterizations of relative utilitarianism in this setting are \textcite{dhillon1999relative}, \textcite{borgers2017revealed} and \textcite{marchant2019utilitarianism}. \textcite{dhillon1999relative} implicitly assumes menu independence and characterizes an aggregation rule that is the same as the one we obtain when relaxing IIA. They have a bottom-up approach, as their implicit objective is to present the weakest possible axioms that characterize relative utilitarianism. In contrast, we have a top-down approach, as we stay close to \textcite{arrow1963social} and wish to present the strongest axioms that don't result in an impossibility. These two different approaches result in distinct axiomatic systems. \textcite{dhillon1999relative} has a higher level of technical sophistication, which is exemplified by their continuity axiom and intricate proof. In contrast, our axioms have a straightforward interpretation, and our proof is much simpler. Note also that \textcite{dhillon1999relative} does not provide a representation when some individuals have either identical or opposing preferences. \textcite{borgers2017revealed} takes a revealed preference approach and imposes axioms on the planner's revealed marginal rate of substitution between individuals' normalized utilities. \textcite{marchant2019utilitarianism} axiomatizes a relative utilitarian choice correspondence rather than a social preference.

In a \textit{single-profile} setting, meaning without axioms relating different profiles, relative utilitarianism has been axiomatized by \textcite{karni1998impartiality}, \textcite{segal2000let}, and \textcite{karni2024impartiality}.

We build on a large body of literature that deals with the aggregation of vNM preferences. The seminal paper is \textcite{harsanyi1955cardinal}, which shows that vNM rationality of society and a Pareto axiom implies linear aggregation of utilities. Many extensions and variations of this theorem followed \autocite{domotor1979ordered, fishburn1984harsanyi, border1985more, selinger1986harsanyi, coulhon1989sociala, hammond1992harsanyi, weymark1993harsanyi, mongin1994harsanyi, weymark1994harsanyi, demeyer1995note, zhou1997harsanyi, blackorby1999harsanyi, mandler2005harsanyi}. While \textcite{harsanyi1955cardinal} shows that social welfare must be a weighted sum of individual utilities, he does not identify the weights. We start from his theorem and impose additional axioms to pin down these weights. Another paper on the aggregation of vNM preferences worthy of mention is \textcite{sprumont2013relative}, which drops the assumption that social preferences are vNM and characterizes the relative egalitarian aggregation rule. Our axiom \ref{ch1:MICA} is quite close to the ``independence of inessential expansions'' axiom in \textcite{sprumont2013relative}.

Closely related to the aggregation of risk preferences is the aggregation of subjective expected utility (SEU) preferences over uncertain acts. Since beliefs regarding the state can differ across individuals, not only conflicting tastes but also conflicting probability judgments have to be aggregated. This additional complication leads to an impossibility result, even in the single-profile setting. \textcite{mongin1995consistent, mongin1998paradox} shows that society's preferences cannot simultaneously satisfy the Pareto condition and have a SEU representation. To escape this impossibility, utilitarian aggregation rules either weaken the rationality postulates underlying the SEU representation \autocite{sprumont2019relative} or the Pareto condition \autocite{gilboa2004utilitarian, brandl2021beliefaveraging, dietrich2021fully}. Most closely related is \textcite{sprumont2019relative}, which allows the social preferences to depend on some reference set of relevant outcomes. This is mathematically similar to our approach, where social preferences can depend on the menu. However, Sprumont's axioms don't tell us what is relevant, as the relevant set is exogenous and part of the social choice problem. Without additional restrictions, this would allow the planner to justify nearly any decision by handpicking the relevant set as needed. One could not simply reinterpret the reference set as the menu, as it is a set of outcomes, not a set of acts.

Related to relative utilitarian aggregation rules is the literature on \textit{range voting} \autocite{smith2000range, gaertner2012general, pivato2014formal, mace2018voting}. Under range voting, a voter assigns a score between 0 and 1 to each candidate, and the candidate with the highest sum of scores is selected as the winner.

We also mention a few of the many other papers on axiomatizations of aggregation rules that start similarly to us from \textcite{arrow1950difficulty,arrow1963social} and relax axioms therein. Notably, \textcite{sen1993internal} drops MI and shows that impossibility still follows if IIA is replaced by \textit{binary IIA}.\footnote{Binary IIA says that the social preference between any two alternatives $x$ and $y$ can only depend on individual preferences between $x$ and $y$. The literature often uses the notion of IIA and binary IIA interchangeably. Note however that binary IIA is only implied by IIA if MI is assumed as well.} To our knowledge, \textcite{sen1993internal} is the only previous paper that points out the implicit MI axiom in \textcite{arrow1950difficulty,arrow1963social}. \textcite{saari1998connecting} and \textcite{maskin2022arrow} relax IIA by allowing the comparison between two alternatives to depend on how many other alternatives lie in between them. 

\section{Axiomatization}\label{ch1:secaxioms}

There is a society consisting of $n$ individuals, where $n\in\mathbb{N}$. The set of individuals is denoted by $N:=\{1,...,n\}$. Furthermore, there is a set of possible alternatives $A$, where $A$ is finite.\footnote{In Appendix \ref{ch1:appendix infinite alternatives} we consider the case where $A$ is infinite.} Each individual in society has rational preferences over the possible alternatives. Formally, let $\triangle A$ denote the set of lotteries over $A$ and let $\mathcal{R}$ denote the set of logically possible von Neumann-Morgenstern (vNM) preferences over $\triangle A$. Individual preferences are then captured by a preference profile $R\in \mathcal{R}^n$. For a given profile $R$, we sometimes write $\succcurlyeq_i^R$ to denote the $i$'th element in $R$. In any given situation, only a subset $S$ of the alternatives in $A$ is feasible and could be implemented by society. We call $S$ the \textit{menu}. For any $S\subseteq A$ and $R\in \mathcal{R}^n$, we call $(S,R)$ the society's \textit{state}, and we denote by $\Omega$ the set of all logically possible states. An \textit{aggregation rule} $\succcurlyeq_*$ assigns to each state $(S,R) \in \Omega$ a binary relation $\succcurlyeq_*^{(S,R)}$ over $\triangle S$. For exposition purposes we act as if there was a social planner that employs an aggregation rule to evaluate the alternatives in the menu. Hence, for a given state $(S,R)$ we refer to $\succcurlyeq_*^{(S,R)}$ as the planner's \textit{evaluation}.

We now impose axioms on how the social planner evaluates alternatives. \textit{Rationality} says that the evaluation is rational in the sense of abiding by the vNM axioms.

\begin{axiomp}{RA}[Rationality]\label{ch1:RA}
For each $(S,R)\in \Omega$, $\succcurlyeq_*^{(S,R)}$ satisfies the vNM axioms.
\end{axiomp}

Rationality is normatively desirable and strengthens Arrow's requirement that the planner's evaluation is complete and transitive. Furthermore, we believe that, since individuals are assumed to be rational, an aggregation rule should preserve this characteristic of the individuals. 

Our second axiom says that the social planner is benevolent, such that the evaluation respects the individuals' preferences whenever these are not in conflict.

\begin{axiomp}{SP}[Strong Pareto]\label{ch1:SP}
For each $(S,R)\in \Omega$ and $x,y\in \triangle S$, if $x \succcurlyeq^R_i y$ for all $i\in N$, then $x \succcurlyeq^{(S,R)}_* y$. If, in addition, $x \succ^R_i y$ for some $i\in N$, then $x \succ^{(S,R)}_* y$.
\end{axiomp}

\ref{ch1:SP} strengthens Arrow's Pareto condition.

Our third axiom is \textit{anonymity}. Anonymity says that the planner's evaluation must not depend on the individual identities but only on the preferences themselves. Hence, in a counterfactual world, where preferences are interchanged across the individuals, the planner's evaluation must be the same.

\begin{axiomp}{AN}[Anonymity]\label{ch1:AN}
For each $(S,R),(S,R')\in \Omega$, if $R'$ is a permutation of $R$, then ${\succcurlyeq^{(S,R)}_*} = {\succcurlyeq^{(S,R')}_*}$.
\end{axiomp}

\ref{ch1:AN} is an impartiality requirement and strengthens Arrow's non-dictatorship axiom.

Above we have stated our desiderata. Next, we will present the two conditions in \textcite{arrow1950difficulty,arrow1963social} that restrict what information can be used in the evaluation. The first condition is \textit{independence of irrelevant alternatives}. We say that two binary relations $\succcurlyeq$ and $\succcurlyeq'$ \textit{agree} on some set of alternatives $S$ if for any $x,y \in \triangle S$, $x \succcurlyeq y$ if and only if $x \succcurlyeq' y$. Furthermore, we say that two preference profiles $R,R' \in \mathcal{R}^n$ agree on $S$ if for each $i\in N$, $\succcurlyeq^R_i$ and $\succcurlyeq^{R'}_i$ agree on $S$.

\begin{axiomp}{IIA}[Independence of Irrelevant Alternatives]\label{ch1:IIA}
Fix $(S,R) \in \Omega$. For any $R'\in \mathcal{R}^n$, if $R$ and $R'$ agree on $S$, then ${\succcurlyeq^{(S,R)}_*}={\succcurlyeq^{(S,R')}_*}$.
\end{axiomp}

IIA says that the planner's evaluation of the menu cannot depend on individual preferences over alternatives outside the menu. Hence, in a counterfactual world, where individual preferences differ only on alternatives outside the menu, the planner's evaluation must be the same.

Arrow's second condition is that social preferences are menu-independent. Arrow assumes this implicitly, as he writes that the social choice from a menu $S$ is made based on a single preference relation over $A$, which is independent of the menu. In light of the numerous ways to formalize this notion, we have selected the following axiom, as it best aligns with our subsequent relaxation of the condition.

\begin{axiomp}{MI}[Menu Independence]\label{ch1:MI}
For each $(S,R)\in \Omega$ and $S'\subseteq S$, $\succcurlyeq^{(S,R)}_*$ and $\succcurlyeq^{(S',R)}_*$ agree on $S'$.
\end{axiomp}

Menu independence says that removing alternatives from the menu does not change the planner's evaluation of the remaining alternatives.

It is well known that Arrow's axioms lead to an impossibility. Unsurprisingly, as our first three axioms strengthen Arrow's rationality, benevolence and impartiality requirements, the above axioms lead to an impossibility as well.\footnote{Note that formally, we weaken Arrow's \textit{universal domain} condition, by assuming that individuals have vNM preferences. However, it has been shown that such a domain restriction is insufficient for escaping the impossibility. See \textcite{sen1970interpersonal}, \textcite{kalai1977aggregation}, \textcite{hylland1980aggregation}, \textcite{chichilnisky1985neumannmorgenstern}, and \textcite{dhillon1997impossibility}.}

\begin{proposition} \label{ch1:arrow}
There is no aggregation rule $\succcurlyeq_*$ that satisfies \ref{ch1:RA}, \ref{ch1:SP}, \ref{ch1:AN}, \ref{ch1:IIA} and \ref{ch1:MI}.
\end{proposition}

As we have argued in the introduction, we believe that both \ref{ch1:IIA} and \ref{ch1:MI} force the social planner to ignore valuable context and should therefore be reconsidered. We will weaken each of these conditions, using the following notion of comparability. Let $[a]$ denote the lottery that assigns probability 1 to the alternative $a\in A$.

\begin{definition} $a\in A$ is \textit{comparable} relative to $B \subseteq A$ under $R\in\mathcal{R}^n$ if (i) $a\notin B$ and (ii) for every $i\in N$, there exists $x_i\in\triangle B$ such that $x_i \sim^R_i [a]$.
\end{definition}

An alternative is comparable to a set if, for each individual, there is a pay-off in the set equal to that of the alternative.

First, we weaken \ref{ch1:MI}.

\begin{axiomp}{MICA}[Menu Independence of Comparable Alternatives]\label{ch1:MICA}
For each $(S,R)\in \Omega$ and $S'\subseteq S$ where every $a\in S\setminus S'$ is comparable relative to $S'$, $\succcurlyeq^{(S,R)}_*$ and $\succcurlyeq^{(S',R)}_*$ agree on $S'$.
\end{axiomp}

MICA says that removing comparable alternatives from the menu does not change the planner's evaluation of the remaining alternatives.

Weakening \ref{ch1:MI} to \ref{ch1:MICA} results in a representation of the planner's evaluation we refer to as \textit{menu-contingent utilitarianism}. For any $R\in \mathcal{R}^n$ and $B \subseteq A$, let $u^{B,R}_i$ denote the representation of $\succcurlyeq^R_i$ where $\max_{a\in B} u_i^{B,R}(a) = 1$ and $\min_{a\in B} u_i^{B,R}(a) = 0$, unless $\succcurlyeq^R_i$ is indifferent on $B$ in which case $u_i^{B,R}(a) = 0$ for all $a\in B$.\footnote{For any binary relation $\succcurlyeq$ on a set $X$, a utility function $u:X\to \mathbb{R}$ is said to \textit{represent} $\succcurlyeq$ if for all $x,y\in X$, $u(x)\geq u(y)$ if and only if $x \succcurlyeq y$.} Furthermore, for any $B\subseteq A$ we denote by $|B|$ the number of elements in $B$.

\begin{theorem}[Menu-Contingent Utilitarianism]\label{ch1:MCE}
Let $|A|\geq 2n+4$. An aggregation rule $\succcurlyeq_*$ satisfies \ref{ch1:RA}, \ref{ch1:SP}, \ref{ch1:AN}, \ref{ch1:IIA} and \ref{ch1:MICA} if and only if for each $(S,R)\in \Omega$, $\succcurlyeq^{(S,R)}_*$ is represented by \begin{equation}\label{ch1:MCErep}\sum_{i\in N} u_i^{S,R}.\end{equation}
\end{theorem}

We sketch the proof of Theorem \ref{ch1:MCE} in Section \ref{ch1:secsketch}. A complete proof can be found in Appendix \ref{ch1:appendix main proof}. Note that no axiom of Theorem \ref{ch1:MCE} is implied by the other axioms of the theorem. Hence, a subset of the axioms would not suffice for the representation. We prove this in Appendix \ref{ch1:appendix droping axioms}.

Next, we weaken \ref{ch1:IIA}. For any $B\subset A$ we write $B^{c}$ to denote $A\setminus B$.

\begin{axiomp}{IICA}[Independence of Irrelevant Comparable Alternatives]\label{ch1:IICA}
Fix $(S,R)\in \Omega$ and $C\subseteq S^{c}$ such that every $a\in C$ is comparable relative to $C^c$ under $R$. For any $R'\in \mathcal{R}^n$, if $R$ and $R'$ agree on $C^c$ and every $a\in C$ is comparable relative to $C^c$ under $R'$, then ${\succcurlyeq_*^{(S,R')}} = {\succcurlyeq_*^{(S,R)}}$.
\end{axiomp}

IICA says that the planner's evaluation of the menu cannot depend on individual preferences over comparable alternatives outside the menu. Hence, in a counterfactual world, where individual preferences over these alternatives are different in a way such that these alternatives are still comparable, the planner's evaluation must be the same.

Weakening \ref{ch1:IIA} to \ref{ch1:IICA} results in a representation we refer to as \textit{setting-contingent utilitarianism}.

\begin{theorem}[Setting-Contingent Utilitarianism]\label{ch1:SCE}
Let $|A|\geq 2n+4$. An aggregation rule $\succcurlyeq_*$ satisfies \ref{ch1:RA}, \ref{ch1:SP}, \ref{ch1:AN}, \ref{ch1:IICA} and \ref{ch1:MI} if and only if for each $(S,R)\in \Omega$, $\succcurlyeq^{(S,R)}_*$ is represented by \begin{equation}\label{ch1:SCErep}\sum_{i\in N} u_i^{A,R}.\end{equation}
\end{theorem}

We prove Theorem \ref{ch1:SCE} in Appendix \ref{ch1:appendix main proof}. No axiom of Theorem \ref{ch1:SCE} is implied by the other axioms of the theorem, which we show in Appendix \ref{ch1:appendix droping axioms}.

In both theorems, we assume that there are at least $2n+4$ possible alternatives. We make this assumption, as it allows us to employ simple and insightful proofs. However, as shown by the following proposition, our proofs require only a few more alternatives than what is necessary for the axioms to be sufficient.

\begin{proposition}\label{ch1:propsuff}
Let $|A|< 2n+1$. There exists an aggregation rule $\succcurlyeq_*$ which satisfies \ref{ch1:RA}, \ref{ch1:SP}, \ref{ch1:AN}, \ref{ch1:IIA} and \ref{ch1:MICA} and which is not represented by (\ref{ch1:MCErep}). Similarly, there exists an aggregation rule $\succcurlyeq_*$ which satisfies \ref{ch1:RA}, \ref{ch1:SP}, \ref{ch1:AN}, \ref{ch1:IICA} and \ref{ch1:MI} and which is not represented by (\ref{ch1:SCErep}).
\end{proposition}

We prove Proposition \ref{ch1:propsuff} in Appendix \ref{ch1:appendix main proof}.

Finally, we want to mention \textit{neutrality}, an additional desideratum that is satisfied by both our representations. Neutrality says that the labels of alternatives play no role in the planner's evaluation. Hence, if the labels of $a$ and $b$ were interchanged and $a$ was preferred by the planner before, then $b$ must be preferred afterward. We believe this to be the most important impartiality requirement besides anonymity. If neutrality wasn't implied by our other axioms, we would have imposed it directly. Since we will refer to neutrality in the upcoming sections, a formal definition is in order. We say that $\pi:A\mapsto A$ is a permutation of $A$ if $\pi$ is bijective and denote by $\Pi$ the set of permutations of $A$. We abuse notation and define $\pi(S):=\left\{\pi(a): a\in S \subseteq A\right\}$. We write $\pi(x) \in \triangle A$ to denote the lottery that for every $a\in A$ assigns probability $\mu\in[0,1]$ to $\pi(a)$ if and only if $x$ assigns probability $\mu$ to $a$. Let $\pi(R):=(\succcurlyeq^{\pi(R)}_i)_{i\in N} \in \mathcal{R}^n$ denote the preference profile which has the same preferences on the permuted alternatives as $R$ on the original alternatives. Formally, $\pi(x) \succcurlyeq_i^{\pi(R)} \pi(y)$ if and only if $x \succcurlyeq^R_i y$ for all $i\in N$ and $x,y\in\triangle A$. 

\begin{axiomp}{NE}[Neutrality]\label{ch1:NE}
For each $(S,R)\in \Omega$, $x,y\in \triangle S$ and $\pi\in\Pi$, $x \succcurlyeq^{(S,R)}_* y$ if and only if $\pi(x) \succcurlyeq^{(\pi(S),\pi(R))}_* \pi(y)$.
\end{axiomp}

\subsection{Discussion}

Our objective was not to identify the weakest set of axioms leading to either of the above representations. Instead, we wanted to demonstrate that even with minor modifications to Arrow's axioms, preference aggregation becomes possible. To make this point convincingly, it is crucial that our axioms are as strong as they can be. That being said, whether weaker axioms would suffice remains an important question. This holds especially true in view of the existing literature that weakens \ref{ch1:IIA} to \textit{independence of redundant alternatives} (\ref{ch1:IRA}); see \textcite{dhillon1999relative} and \textcite{brandl2021beliefaveraging}. A redundant alternative is one for which there is a lottery over the remaining alternatives that leaves every individual indifferent between the two. While any redundant alternative is comparable, the converse does not hold. Consequently, \ref{ch1:IRA} is weaker than \ref{ch1:IICA}. We show in Appendix \ref{ch1:appendix minimal sufficiency} that weakening \ref{ch1:IICA} to \ref{ch1:IRA} would not suffice for our representation. Similarly, weakening \ref{ch1:MICA} to \textit{menu independence of redundant alternatives} (\ref{ch1:MIRA}) wouldn't suffice either. However, in another aspect \ref{ch1:IICA} is indeed stronger than necessary.  It would suffice if \ref{ch1:IICA} only applied to cases where all alternatives outside the menu are comparable (i.e., $C=S^c$), akin to the formulation of \ref{ch1:IRA}. We opted for the stronger version of the axiom as it more closely aligns with Arrow's \ref{ch1:IIA}.

\section{Sketch of Proof}\label{ch1:secsketch}

To highlight some insights, this section provides a sketch of the proof. Since the proofs of both theorems are quite similar, we only sketch the proof of Theorem \ref{ch1:MCE}.

Before we prove the theorem, we derive two interim results. Note that \ref{ch1:SP} implies \textit{Pareto indifference} (\ref{ch1:PI}), which says that if every individual in society is indifferent between two lotteries, then so is the social planner. The first interim result states that if both \ref{ch1:RA} and \ref{ch1:PI} are satisfied, then the utility function of the social planner can be expressed as a weighted sum of the individual utility functions. This result is well known and has first been postulated by \textcite{harsanyi1955cardinal}. However, Harsanyi's original proof contains a mistake, which led to a variety of proofs by the subsequent literature. We, in turn, provide our own proof of Harsanyi's theorem, similar to those by \textcite{border1985more}, \textcite{selinger1986harsanyi}, and \textcite{hammond1992harsanyi}, albeit ours is self-contained as we do not refer to mathematical theorems. Our proof makes use of the \textit{pay-off matrix}, which indicates for each individual the vNM utility of every alternative in the menu.\footnote{A row of 1's is included in the pay-off matrix to allow for a constant in the representation.} Figure \ref{ch1:fig:pay} shows an example for three individuals and four alternatives.\begin{figure}[h]
\center
$\kbordermatrix{
     & a & b & c & d\\
     & 1 & 1 & 1 & 1 \\
    i=1 & .5 & .5 & 0 & 1 \\
    i=2 & 1 & 0 & .7 & .3 \\
    i=3 & 0 & 1 & .2 & .8
  }$
\caption{Pay-off matrix.}
\label{ch1:fig:pay}
\end{figure} By \ref{ch1:RA}, the planner's evaluation can be represented by a row vector of vNM utilities as well. What needs to be shown is that whenever \ref{ch1:PI} is satisfied, the planner's utility vector is equal to some linear combination of the rows in the pay-off matrix. If individual preferences are sufficiently diverse, such that the rows of the pay-off matrix span the entire vector space, then any logically possible vNM preference over the menu can be expressed by a linear combination of the individual utility functions. If individual preferences are not sufficiently diverse, we show that there exists a \textit{dependent alternative}, whose column can be expressed by a linear combination of the other columns. For instance, in our example Column $d$ is equal to $a$ plus $b$ minus $c$. We then sequentially drop dependent alternatives until the columns of the remaining alternatives are linearly independent. Then individual preferences over the remaining alternatives are sufficiently diverse such that any vNM preference over these alternatives can be expressed by a linear combination of rows. For each of the dependent alternatives, we can identify two lotteries, one of them involving the dependent alternative, such that every individual is indifferent between them. In the case of alternative $d$, these lotteries would be $\frac{1}{2}[a]+\frac{1}{2}[b]$ and $\frac{1}{2}[c]+\frac{1}{2}[d]$. By \ref{ch1:PI}, also the planner must be indifferent between these lotteries, and hence, the planner's utility of the dependent alternative is pinned down by the same linear combination as for the individuals.

For our second interim result, we introduce the concept of \textit{polar} alternatives. An alternative is polar if it is best among the menu for one individual and worst among the menu for everyone else. If the menu consists only of polar alternatives, we call such a state a \textit{polar state}. Our second interim result then says that in a polar state, the planner is indifferent between all alternatives. Consider, for instance, a polar state as described by the pay-off matrix depicted on the left-hand side of Figure \ref{ch1:fig:polar}. \begin{figure}[h]
\center
\begin{align}\nonumber
	\kbordermatrix{
     & e & f & g\\
     & 1 & 1 & 1 \\
    i=1 & 1 & 0 & 0 \\
    i=2 & 0 & 1 & 0 \\
    i=3 & 0 & 0 & 1
  } && {\large \xrightleftharpoons[\text{neutrality}]{\text{anonymity}}} && \kbordermatrix{
     & e & f & g\\
     & 1 & 1 & 1 \\
    i=1 & 0 & 1 & 0 \\
    i=2 & 1 & 0 & 0 \\
    i=3 & 0 & 0 & 1
  }
\end{align}
\caption{Polar state.}
\label{ch1:fig:polar}
\end{figure} Note that the pay-off matrix doesn't describe the state completely, as it does not specify individual preferences over infeasible alternatives. However, by \ref{ch1:IIA}, preferences over infeasible alternatives can be ignored. We assume that the planner weakly prefers $e$ over $f$ and show that this implies a weak preference for $f$ over $e$. This, of course, only leaves indifference. First, notice that the pay-off matrix on the right-hand side of Figure \ref{ch1:fig:polar} depicts the state that results from permuting the labels of Individuals 1 and 2 in the left-hand state. If $e$ is weakly preferred to $f$ in the left-hand state, by \ref{ch1:AN} this must be the case in the right-hand state as well. Second, notice that the left-hand state results from permuting the labels of alternatives $e$ and $f$ in the right-hand state. If $e$ is weakly preferred to $f$ in the right-hand state, \ref{ch1:NE} would demand that $f$ must be weakly preferred to $e$ in the left-hand state. With the help of \ref{ch1:MICA} and \ref{ch1:PI}, we show that indeed \ref{ch1:NE} is satisfied in polar states. We leave the details to the formal proof in the appendix.

Finally, we prove the theorem. Consider individual preferences over the menu $\{a,b,c,d\}$ as depicted in Figure \ref{ch1:fig:pay} and assume that for each individual, there is a comparable polar alternative outside the menu, denoted by $e$, $f$, and $g$. In the menu $\{e,f,g\}$, the planner must be totally indifferent as we have shown previously. Now we add the original menu, resulting in Figure \ref{ch1:fig:combine}. \begin{figure}[h]
\center
$\kbordermatrix{
     & a & b & c & d & e & f & g\\
     & 1 & 1 & 1 & 1 & 1 & 1 & 1\\
    i=1 & .5 & .5 & 0 & 1 & 1 & 0 & 0\\
    i=2 & 1 & 0 & .7 & .3 & 0 & 1 & 0\\
    i=3 & 0 & 1 & .2 & .8 & 0 & 0 & 1
  }$
\caption{Resulting pay-off matrix.}
\label{ch1:fig:combine}
\end{figure} Since we have only added comparable alternatives, by \ref{ch1:MICA}, the planner must still be indifferent between $e$, $f$, and $g$. This indifference, together with our first interim result, then implies that equal weights on the individual utility functions represent the planner's evaluation. \ref{ch1:SP} ensures that these common weights are positive and can be normalized to 1. Finally, we remove the polar alternatives, and by \ref{ch1:MICA}, the same linear combination must still represent the planner's evaluation. By \ref{ch1:IIA}, this must hold even if there are no comparable polar alternatives outside the initial menu. Note that if there are fewer than $n$ alternatives outside the initial menu, the proof is more involved.

\section{Applications} \label{ch1:secapplic}

We apply our aggregation rules to three classic economic situations: aggregating the welfare of different consumers in a market, identifying the welfare of a dynamically inconsistent decision maker, and finding a fair solution to a bargaining problem. The literature has so far treated each of these problems in isolation. In the previous sections, we provided a unifying framework that aggregates individual preferences consistently across applications. With this methodology we uncover policy recommendations contrary to those by the established welfare criteria.

\subsection{Consumer Welfare}

Consider a society with $n$ individuals and two consumption goods, $m$ and $g$, where $m$ is the numéraire. A social alternative is an allocation of these goods to the individuals in society, hence an element of $\mathbb{R}_+^{2n}$. Individuals are self-interested and their utility is quasi-linear, formally for each $i\in N$, \begin{equation}\label{ch1:quasilinearut}u_i((m_1,g_1),...,(m_n,g_n))=\alpha_i m_i + v_i(g_i)\end{equation} for some $v_i:\mathbb{R}_+ \to \mathbb{R}$ and $\alpha_i\in\mathbb{R}_+$. The standard measure of aggregate welfare in this setting is \textit{total consumer surplus}, short TCS, given by 
\begin{equation}\nonumber
\begin{split}
W_{\text{TCS}}((m_1,g_1),...,(m_n,g_n)) & =\sum_{i\in N} \frac{1}{\alpha_i} u_i((m_1,g_1),...,(m_n,g_n))\\
& = \sum_{i\in N} m_i + \sum_{i\in N} \frac{v_i(g_i)}{\alpha_i}.
\end{split}
\end{equation}
Consider two allocations, $a$ and $b$, such that some individuals strictly prefer $a$ and some strictly prefer $b$. If $W_{\text{TCS}}(a)$ is higher than $W_{\text{TCS}}(b)$, then there exist monetary transfers after $a$, resulting in an allocation $a'$, such that $a'$ Pareto dominates $b$. This makes TCS quite appealing. However, if these transfers aren't implemented, TCS makes a judgment on how the well-being of different individuals should be traded off. We show that the way in which these trade-offs are made violates neutrality. Specifically, TCS is sensitive to which of the two consumption goods is selected as the numéraire. Let $\mathcal{R}^{QL}$ be the subset of $\mathcal{R}^n$ such that individual preferences can be represented by quasi-linear utility functions as in (\ref{ch1:quasilinearut}).

\begin{proposition}\label{ch1:consumprop}
There doesn't exist an aggregation rule $\succcurlyeq_*$ on the restricted domain $\mathcal{R}^{QL}$ that is represented by $W_{\text{TCS}}$ and satisfies \ref{ch1:NE}.
\end{proposition}

\begin{proof}
It is sufficient to show that the axiom is violated in some state, so consider a state $R\in \mathcal{R}^{QL}$ where for each $i\in N$, $v_i(g_i)=\beta_i g_i$ for some $\beta_i>0$. Next consider the permutation $\pi$ of alternatives such that $\pi((m_1,g_1),...,(m_n,g_n)) = ((g_1,m_1),...,(g_n,m_n))$ and note that $\pi(R)\in \mathcal{R}^{QL}$. For \ref{ch1:NE} to be satisfied, it must hold that
\begin{equation}\label{ch1:csneutralproof}W_{\text{TCS}}^{R}(a) \geq W_{\text{TCS}}^{R}(b)\text{ if and only if }W_{\text{TCS}}^{\pi(R)}(\pi(a)) \geq W_{\text{TCS}}^{\pi(R)}(\pi(b)).\end{equation} Let $a=((0,1),(0,0),...,(0,0))$ and $b=((0,0),(2,0),...,(0,0))$. Then $W_{\text{TCS}}^{R}(a)=\frac{\beta_1}{\alpha_1}$, $W_{\text{TCS}}^{R}(b)=2$, $W_{\text{TCS}}^{\pi(R)}(\pi(a))=1$ and $W_{\text{TCS}}^{\pi(R)}(\pi(b))=\frac{\alpha_2}{\beta_2}$. This violates (\ref{ch1:csneutralproof}) for instance when $\beta_1=3$, $\alpha_1=1$, $\alpha_2=2$ and $\beta_2=1$.
\end{proof}

Now assume that allocations result from the following setting. A monopolist produces the good $g$ at constant marginal cost $c$. Individuals can be divided into two distinct groups: students $s$ and adults $a$. Within each group, every individual has the same preferences. The fraction of individuals belonging to Group $J\in\{s,a\}$ is denoted by $\gamma_J$. Total demand of Group $J$ is then $D_J(p):=n\gamma_J(v'_i)^{-1}(p \alpha_i)$ where $i\in J$. We assume that the monopolist can identify which group an individual belongs to and hence can choose a price pair $(p_s,p_a)\in [c,\infty)^2$, where $p_J$ denotes the price charged to Group $J$. We denote by $\psi(p_s,p_a) \in \mathbb{R}_+^{2n}$ the allocation resulting from a price pair $(p_s,p_a)$. Let $p_J^*$ denote the monopoly price charged to Group $J$, let $p^*$ denote the profit-maximizing price if price discrimination was prohibited, and assume $$p^*_a> p^* > p^*_s.$$ A policymaker can choose whether to allow or prohibit price discrimination, hence $S=\{\psi(p^*_s,p^*_a),\psi(p^*,p^*)\}$. To guide this decision, the policymaker wants to assess whether the prohibition of price discrimination benefits consumers. According to TCS, aggregate welfare is given by $$W_{\text{TCS}}(\psi(p_s,p_a))=CS_s(p_s)+CS_a(p_a)$$ where $CS_J(p):=\int_p^\infty D_J(p)dp$. Since total consumer surplus violates neutrality, we propose setting-contingent utilitarianism as an alternative. The setting puts natural bounds on the possible alternatives, namely $A=\{\psi(p_s,p_a) \in \mathbb{R}_+^{2n}:(p_s,p_a)\in [c,\infty)^2\}$. Setting-contingent utilitarianism, short SCU, is then given by $$W_{\text{SCU}}(\psi(p_s,p_a))= 
\frac{n \gamma_s}{CS_s(c)} CS_s(p_s)+ 
\frac{n \gamma_a}{CS_a(c)}CS_a(p_a).$$

We now compare these two welfare measures. To make things simple, we consider the textbook case of linear demand. It is well known that when demand is linear, total consumer surplus is higher under price discrimination if and only if one group isn't served under the uniform price. Our welfare measure, on the other hand, suggests that price discrimination can be socially beneficial, even if both groups are served under uniform pricing. We demonstrate this with the help of the following example. Figure \ref{ch1:fig:cs} shows the demand and monopoly prices for $D_s(p)=2000-50p$, $D_a(p)=7000-100p$, and $c=0$. \begin{figure}[h]
\center
\includegraphics[scale=1.3]{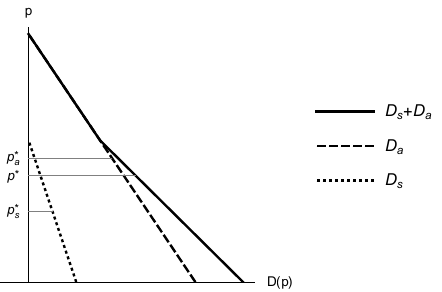}
\caption{Linear demand and monopoly prices.} 
\label{ch1:fig:cs}
\end{figure}
If the firm is allowed to price discriminate, it will charge $\$20$ to students and $\$35$ to adults, which yields a consumer surplus of $CS_s(\$20)=\$10,000$ and $CS_a(\$20)=\$61,250$. If price discrimination were prohibited, the firm would charge $\$30$ to both, which yields a consumer surplus of $CS_s(\$30)=\$ 2,500$ and $CS_a(\$30)=\$ 80,000$. Note that TCS is higher when price discrimination is prohibited. For both groups to be better off under uniform pricing, adults would have to pay between $\$7,500$ and $\$ 18,750$ to students. However, transfers are typically not implemented, which makes students worse off under uniform pricing. Going from price discrimination to uniform pricing, the normalized utility of students decreases from $0.25$ to roughly $0.06$, while it increases for adults from $0.25$ to roughly $0.33$. If both groups consist of an equal number of individuals, then total utility is higher under price discrimination and is therefore favored by SCU. In a sense, SCU penalizes adults through a lower weight on their surplus in order to account for the missing transfers. Note that if transfers would, in fact, be implemented, then SCU would agree with TCS since both satisfy \ref{ch1:SP}. Furthermore, SCU accounts for how many individuals are positively and negatively affected by the policy, whereas TCS ignores the shares of individuals in each group. If, for instance, the share of students was below roughly $0.29$, then SCU would favor the prohibition of price discrimination, as sufficiently many adults would be positively affected by the policy. SCU values individuality and has a flavor similar to the majority principle. Finally, note that SCU, unlike TCS, can be meaningfully applied even when individual utilities are not quasi-linear.

\subsection{Intertemporal Welfare}

Consider a finitely lived decision maker (DM) who consumes a good in each period up to period $n$. As in \textcite{laibson1997golden}, we assume the DM discounts consumption quasi-hyperbolically, such that the utility of a consumption sequence $(c_1,...,c_n)\in \mathbb{R}_+^{n}$ experienced in period $t$ is given by
\begin{equation}\label{ch1:quasihyper}
u_t(c_t,...,c_n)=v(c_t)+\beta\sum_{i=t+1}^n\delta^{i-t} v(c_i)
\end{equation}
for $\beta,\delta\in[0,1]$ and some monotone per-period valuation $v:\mathbb{R}_+ \to \mathbb{R}$. If $\beta<1$, the DM is dynamically inconsistent, meaning there exist sequences $(c_1,...,c_n)$, $(c_1,...,c_t,c'_{t+1},...,c'_n)\in \mathbb{R}_+^{n}$ such that the DM in Period $t$ strictly prefers $(c_1,...,c_n)$ over $(c_1,...,c_t,c'_{t+1},...,c'_n)$, while in Period $t+1$ her preference is reversed. It is as if the DM consists of different \textit{selves}, Self 1 to Self $n$, with conflicting interests. If we want to assess which of these two sequences is better for the DM overall, we need a welfare criterion that incorporates the perspective of each self. Two criteria are common in the literature, the Pareto criterion and \textit{long-run utility} \autocite{odonoghue1999doing} given by $$W_{\text{LRU}}(c_1,...,c_n):=\sum_{t=1}^n\delta^{t-1} v(c_i).$$

Before we discuss these criteria, let us first consider the case where $\beta=1$, such that the DM is dynamically consistent. It is conventional wisdom among economists that, in this case, there is no conflict between the selves, and the appropriate welfare measure is the utility of Self 1. However, this view ignores the possibility that, even though there is no conflict looking forward, there could be a conflict looking backward. For instance, the DM might regret an earlier decision and prefer they had saved more in the past while simultaneously agreeing with the earlier self on how much to save today. To formalize this idea, let each self $t$ have a vNM preference $\succcurlyeq_t$ over $\triangle (\mathbb{R}_+^{n})$. This means that Self $t$ can compare sequences that differ in the consumption levels before Period $t$. However, note that since one cannot affect one's past, $\succcurlyeq_t$ is only partially revealed. Let $\mathcal{R}^{QH}$ be the subset of $\mathcal{R}^{n}$ such that the revealed part of the individual preference can be represented by quasi-hyperbolically discounted utility as in (\ref{ch1:quasihyper}). Formally, for any $(\succcurlyeq_1,...,\succcurlyeq_n)\in \mathcal{R}^{QH}$, $(c_1,...,c_n)\in \mathbb{R}_+^{n}$ and $t\in \{1,...,n\}$, preferences of Self $t$ over lotteries that assign probability 1 on sequences starting with $(c_1,...,c_{t-1})$ can be represented by (\ref{ch1:quasihyper}). Note that the selves are in conflict unless $\succcurlyeq_i = \succcurlyeq_j$ for all $i,j\in\{1,...,n\}$. Hence, for all selves to agree, the DM would have to value past consumption more than current consumption and value it more the further it lies in the past. Since this doesn't seem very plausible, we would expect backward-looking disagreements to be common. Now that all selves have preferences over the same domain, we can view these conflicting interests from the perspective of preference aggregation. This enables us to demonstrate that taking Self 1's utility as a measure of total welfare whenever $\beta=1$ violates our Pareto condition.

\begin{proposition}\label{ch1:firstselfprop}
Let $\succcurlyeq_*$ be an aggregation rule on the restricted domain $\mathcal{R}^{QH}$, such that $\succcurlyeq^{(S,R)}_*=\succcurlyeq^R_1$ whenever $\beta=1$. Then $\succcurlyeq_*$ violates \ref{ch1:SP}.
\end{proposition}

\begin{proof}
Note that it is sufficient to show that the axiom is violated in some state. Consider the state where each self is \textit{past indifferent}, meaning for any $t\in\{2,...,n\}$, $(c_1,...,c_n),(c'_1,...,c'_{t-1},c_t,...,c_n)\in \mathbb{R}_+^{n}$, $$(c_1,...,c_n) \sim_t (c'_1,...,c'_{t-1},c_t,...,c_n).$$ Now consider two sequences $s=(c_1,...,c_n)$ and $s'=(c'_1,c'_2,c_3...,c_n)$ such that $c_1>c'_1$ and $v(c_1)+\delta v(c_2) = v(c'_1)+\delta v(c'_2)$. Then $s \sim_1 s'$, $s'\succ_2 s$ and $s \sim_t s'$ for all $t\geq 3$ and hence by \ref{ch1:SP} $s'$ should give strictly higher total welfare than $s$. If, however, $u_1$ measures total welfare of the DM, then $s$ and $s'$ are equally desirable.
\end{proof}

Let us now return to the aforementioned welfare criteria. The Pareto principle cannot compare every sequence in $\mathbb{R}_+^{n}$; hence, is incomplete and violates \ref{ch1:RA}. Long-run utility is equal to $u_1$ for $\beta=1$, and therefore, as shown by Proposition \ref{ch1:firstselfprop}, violates \ref{ch1:SP}. As an alternative to the established criteria, we propose menu-contingent utilitarianism. We assume that the social planner can distribute an initial endowment of size 1 across periods, hence $S=\{(c_1,...,c_n)\in [0,1]^n:\sum_1^n c_t \leq 1\}$. Furthermore, in order to apply our criterion, we have to make assumptions on individual preferences over histories. Note that the same is true for the Pareto principle (see \textcite{goldman1979intertemporally}). As a benchmark, we assume past indifference. Then, preferences over entire sequences are represented by (\ref{ch1:quasihyper}). Next, we have to normalize each self's utility with respect to the best and worst alternative in $S$. Without loss of generality, assume that $v(0)=0$ and let $\bar u_t=\max_{s\in S} u_t(s)$ denote the optimal allocation from the perspective of Self $t$. Then, total welfare of the DM according to menu-contingent utilitarianism is given by
$$W_{\text{MCU}}(c_1,...,c_n):=\sum_{t=1}^n \frac{1}{\bar u_t}u_t(c_t,...,c_n)=\sum_{t=1}^n \underbrace{\left(\frac{1}{\bar u_t} + \beta \sum_{i=1}^{t-1} \frac{\delta^i}{\bar u_{t-i}} \right)}_{\gamma_t} v(c_t).$$ In contrast to long-run utility, the weights $\gamma_t$ on the per period valuation are increasing in $t$. Hence, our criterion recommends an increasing consumption profile. This is because future consumption has a positive externality on earlier selves in the form of anticipatory utility, while later selves do not benefit from past consumption. Of course, this is driven by our assumption of past indifference. We leave it to future research to determine whether this assumption is plausible. Finally, note that MCU can be applied even when the DM is not a quasi-hyperbolic discounter.

\subsection{Bargaining}

Consider $n$ individuals who have the possibility to cooperate and create a surplus. In order to generate this surplus, they must agree on an alternative to be implemented. For instance, one could think of a buyer and seller bargaining over the price or a worker and an employer negotiating a wage. If the group cannot reach an agreement, each individual resorts to their respective outside options, which we call the \textit{disagreement alternative} and denote by $a_{0}$. In the following exposition, we consider bargaining through the eyes of an arbitrator who, to facilitate cooperation, has to choose an alternative for the group. Equivalently, one could phrase the selection as a consequence of fairness postulates.

The literature on axiomatic bargaining suggests that the arbitrator adheres to the following approach. The problem is first reformulated in terms of vectors of utilities, and then a bargaining solution is applied. This reformulation is discussed in more detail below. Here, we recall the bargaining approach. A \textit{bargaining problem} $(d,U)$ specifies a set of utility vectors $U \subset \mathbb{R}^n$, called the \textit{bargaining set}, and a \textit{disagreement point} $d=(d_1,...,d_n)\in U$. We assume that $U$ is compact and denote by $\mathcal{B}$ the set of all such bargaining problems. A bargaining solution $f$ assigns to each $(d,U)$ a subset of $U$, hence $f(d,U)\subseteq U$. Two bargaining solutions are prominent in the literature: the \textit{Nash bargaining solution} \autocite{nash1950bargaining}, short Nash, and the \textit{Kalai-Smorodinsky bargaining solution} \autocite{kalai1975other}, short KS. Note that both solutions require the bargaining set to be convex. Denote by $\mathcal{B}_{\text{con}}$ the subset of $\mathcal{B}$ for which $U$ is convex. Then Nash is given by $$f_{\text{Nash}}(d,U):=\argmax_{(v_1,...,v_n)\in U} \prod_{i\in N}\left(v_i-d_i\right)$$ for all $(d,U)\in \mathcal{B}_{\text{con}}$. Next, we define KS, which requires some additional notation. For any $(d,U)\in \mathcal{B}$, let $\Lambda(d,U):=\{(v_1,...,v_n)\in U:v_i\geq d_i\text{ for all }i\in N\}$. Hence, $\Lambda(d,U)$ is the set of utility vectors in $U$ that weakly Pareto-dominate $d$. Furthermore, let $\alpha_i(d,U)$ denote $i$'s maximal utility among all points in $\Lambda(d,U)$. Note that KS only applies for $n=2$. So $f_{\text{KS}}(d,U)$ is the intersection of the Pareto frontier of $U$ and the line connecting $d$ and $(\alpha_1(d,U), \alpha_2(d,U))$ for all $(d,U)\in \mathcal{B}_{\text{con}}$.

We now investigate whether these bargaining solutions could be used by an arbitrator who first evaluates the different alternatives and then selects the best according to this evaluation. Let $A$ denote the set of alternatives and note that $a_{0} \in A$. Let $\mathcal{U}$ denote the set of possible vNM utility functions over $\triangle A$. A utility profile $u\in \mathcal{U}^n$ specifies for each individual a utility function over the possible alternatives. We say that a bargaining problem $(d,U)\in \mathcal{B}$ is \textit{associated} with a utility profile $u\in \mathcal{U}^n$ (and vice versa) if $d=(u_1(a_{0}),...,u_n(a_{0}))$ and $U=\{(u_1(a),...,u_n(a)):a\in A\}.$ To ensure that any bargaining problem in $\mathcal{B}$ can be generated by some utility profile in $\mathcal{U}^n$, we require $A$ to be sufficiently rich. In line with our motivating examples, the reader might think of $A$ as the set of prices or wages, i.e., $A=\mathbb{R}_+$. We denote by $\succcurlyeq_+$ the arbitrator's aggregation rule, such that for every $u\in \mathcal{U}^n$, $\succcurlyeq_+^u$ is a binary relation over $\triangle A$. Note that we deviate from the aggregation rule as defined in Section \ref{ch1:secaxioms} in two ways. We have indicated this by the use of a different subscript. First, for notational convenience, $\succcurlyeq_+$ is a function of utility profiles rather than preference profiles. Second, $\succcurlyeq_+$ does not depend on the menu because, as we will argue later, in a bargaining setting, the menu is determined endogenously. Neither of these deviations is responsible for our subsequent finding. Furthermore, in line with Section \ref{ch1:secaxioms} we require the arbitrator to evaluate lotteries over $A$. We don't insist that randomization is feasible, but if it were feasible, the arbitrator should be able to evaluate the resulting lotteries. Finally, we connect the aggregation rule to a bargaining solution through the following definition.

\begin{definition}\label{ch1:bargain consistent}
$\succcurlyeq_+$ is \textit{consistent} with $f$ if for any $(d,U)\in \mathcal{B}_{\text{con}}$   and associated $u\in \mathcal{U}^n$, the following holds. For all $x\in \triangle A$, $$(u_1(x),...,u_n(x))\in f(d,U)\text{ if and only if }x \in \{y \in \triangle A:y \succcurlyeq_+^u z\text{ for all }z\in \triangle A\}.$$
\end{definition}

The definition captures an implicit assumption of the bargaining approach, namely that alternatives are selected solely based on their utility vectors. Hence, if some lottery's utility vector is chosen by the bargaining solution, this lottery must be among the most preferred. Conversely, if a lottery is among the most preferred, its utility vector must be chosen by the bargaining solution. We have restricted the definition to convex bargaining sets as this is the domain of the prominent bargaining solutions.

We can now investigate whether a given bargain solution is consistent with certain desirable features of our aggregation rule. We find that neither Nash nor KS is consistent with rationality. More generally, any bargaining rule that imposes ex-ante symmetry is inconsistent with the vNM postulates. This is well known; see for instance \textcite{diamond1965evaluation}, \textcite{harsanyi1975nonlinear}, \textcite{broome1984uncertainty} or \textcite{machina1989dynamic}. Nevertheless, we provide the following proposition for the sake of completeness.

\begin{proposition}
There doesn't exist an aggregation rule $\succcurlyeq_+$ that satisfies \ref{ch1:RA} and is consistent with either $f_{\text{Nash}}$ or $f_{\text{KS}}$.
\end{proposition}

\begin{proof}
We prove the proposition for $n=2$ as this is required for KS. The proof can be easily extended to a general $n$. The proof strategy goes as follows. We assume that $\succcurlyeq_+$ is consistent with a bargaining solution $f$ that satisfies Nash's \textit{symmetry} and Pareto axiom. Symmetry says that in a symmetric bargaining problem, any element of the bargaining solution must assign equal utility to all individuals. Note that both Nash and KS satisfy symmetry. We then show that $\succcurlyeq_+$ must violate \ref{ch1:RA}.

Consider the bargaining problem $((0,0),\mathcal{S})$ where $\mathcal{S}$ is the unit simplex. As $((0,0),\mathcal{S})$ is symmetric, $f$ selects $\{(0.5,0.5)\}$. See Figure \ref{ch1:fig:bargproof} for an illustration. \begin{figure}[h]
\center
\includegraphics[scale=0.55]{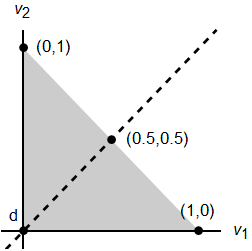}
\caption{Symmetric bargaining solution for $((0,0),\mathcal{S})$.}
\label{ch1:fig:bargproof}
\end{figure}
For any $u\in\mathcal{U}^2$ associated with $((0,0),\mathcal{S})$, there must exist two alternatives $a_1,a_2 \in A$ that generate the extreme points of $\mathcal{S}$, i.e., $(u_1(a_1), u_2(a_1))=(1,0)$ and $(u_1(a_2), u_2(a_2))=(0,1)$. Now consider the coin-flip between $a_1$ and $a_2$, denoted by $x := \frac{1}{2}[a_1] + \frac{1}{2}[a_2]$. As $(u_1( x),u_2( x))=(0.5,0.5)\in f((0,0),\mathcal{S})$, $x$ must be among the most preferred elements according to $\succcurlyeq_+^u$. Formally $x \in \{y \in \triangle A:y \succcurlyeq_+^u z\text{ for all }z\in \triangle A\}$. But neither $a_1$ nor $a_2$ are among the most preferred elements, as neither $(1,0)$ nor $(0,1)$ is selected by $f$. Therefore,
\begin{equation}\label{ch1:bargproof1}\frac{1}{2}[a_1] + \frac{1}{2}[a_2] \succ_+^u [a_1],\end{equation} 
\begin{equation}\label{ch1:bargproof2}\frac{1}{2}[a_1] + \frac{1}{2}[a_2] \succ_+^u [a_2].\end{equation}
By the vNM independence axiom, (\ref{ch1:bargproof1}) would imply $[a_2] \succ_+^u [a_1]$ and (\ref{ch1:bargproof2}) would imply $[a_1] \succ_+^u [a_2]$, a contradiction. Hence, the vNM independence axiom and consequently \ref{ch1:RA} must be violated.
\end{proof}

We believe that the arbitrator should make the decision based on a rational evaluation of the alternatives. As this is not reconcilable with the prominent bargaining solutions, we offer menu-contingent utilitarianism as an alternative. The menu arises naturally in this setting. Because an individual would not agree to an alternative worse than $a_{0}$, it is the set of alternatives that weakly Pareto-dominate $a_{0}$. Formally, given $u\in\mathcal{U}^n$ the menu is $S:=\{a\in A:u_i(a)\geq u_i(a_{0})\text{ for all }i\in N\}$. Note that a bargaining problem $(d,U)$ contains all relevant information for our rule to evaluate the utility vectors in $\Lambda(d,U)$. Hence, menu-contingent utilitarianism leads to a well-defined bargaining solution, which we abbreviated by MCU and which is given by $$f_{\text{MCU}}(d,U):=\argmax_{(v_1,...,v_n)\in \Lambda(d,U)} \sum_{i\in N}\frac{v_i-d_i}{\alpha_i(d,U)-d_i}$$ for all $(d,U)\in \mathcal{B}$. MCU selects the utility vectors that are most preferred under a menu-contingent utilitarian evaluation. Note that MCU does not require convexity of the bargaining set. The same bargaining solution has been axiomatized before by \textcite{pivato2009twofold} under the assumption of convexity of the bargaining set. A similar bargaining solution has been axiomatized by \textcite{cao1982preference} and \textcite{baris2018timing}, again assuming convexity of the bargaining set. Their solution differs from ours, as they normalize individual utilities with respect to each individual's highest utility in $U$, whereas we normalize with respect to each individual's highest utility in $\Lambda(d,U)$.

Like Nash and KS, MCU has a neat geometric interpretation, which we demonstrate with the help of Figure \ref{ch1:fig:barg}. \begin{figure}[h]
\center
\includegraphics[scale=1]{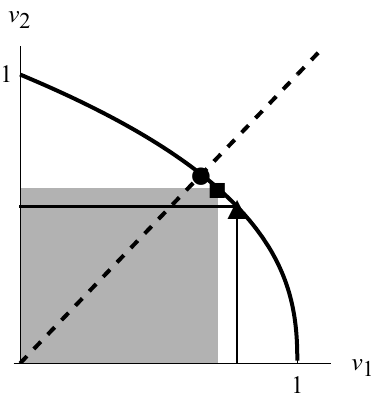}
\caption{Nash, KS and MCU.} 
\label{ch1:fig:barg}
\end{figure} The disagreement point is the origin and the curve shows the Pareto frontier of bargaining set for $u_1(a)=1-a$ and $u_2(a)=a^{\frac{2}{5}}$. KS (circle) is the intersection of the Pareto frontier and the dashed $45 \degree$ line, Nash (square) maximizes the area of the rectangle spanned by $d$ and the Pareto frontier and MCU (triangle) maximizes the circumference of the rectangle spanned by $d$ and the Pareto frontier. If the Pareto frontier is smooth, then the slope of the Pareto frontier at MCU is -1, as it would otherwise be possible to increase the utility of one individual by decreasing the utility of the other individual by a lesser amount. Hence, in this example, MCU favors Individual 1 relative to KS because the slope at KS is flatter than -1 such that Individual 1 can be made significantly better off at the cost of making Individual 2 slightly worse off.

\section{Conclusion} \label{ch1:secconclu}

We initiated this project in pursuit of a rule that is rational, benevolent, and impartial, as these desiderata are not satisfied by the rules that are typically used in applications. We learned from \textcite{arrow1950difficulty, arrow1963social} that for a rule to satisfy these desiderata, it must be sensitive to context. We offer two rules, each motivated by adhering closely to Arrow. Each rule results from slightly weakening one of Arrow's two context-restricting axioms, \ref{ch1:MI} and \ref{ch1:IIA}. Each rule makes an explicit normative judgment about what is allowed to matter. We believe that in different situations, different considerations are allowed to matter. Consequently, we find that there is no single unifying rule. This revelation could pave the way for numerous rules derived from various modifications of \ref{ch1:MI} and \ref{ch1:IIA}. However, if we permit a planner to tailor the context to a specific application, almost any social choice can be justified. To maintain discipline, we recommend selecting from our two rules, as each represents a normative judgment on one side of the spectrum.

The menu-contingent utilitarian rule satisfies \ref{ch1:IIA} and thus reflects the judgment that nothing outside the menu is allowed to matter. However, the menu itself does matter. We believe this normative consideration to be relevant for bargaining situations, as bargaining involves finding a compromise that provides each individual with a fair share of the surplus. Our understanding of fairness is relative to what can be achieved and must disregard alternatives that individuals would never agree upon. We also propose this rule for practical convenience in situations where there is no information about individual preferences outside the menu or when the setting does not impose natural bounds on individual utilities.

On the other side of the spectrum is the setting-contingent utilitarian rule. This rule satisfies \ref{ch1:MI}, hence reflects the judgment that the menu itself is not allowed to matter. Individual preferences over all possible alternatives are taken into account, independently of whether an alternative is currently feasible. We believe this normative consideration to be relevant for policymaking, where the policymaker should take a birds-eye view and assess welfare on the basis of the entire setting. In addition, policymakers often face binary decisions sequentially, namely whether to adopt a specific policy or not. Menu independence ensures that these decisions are consistent across time.

\begin{appendices}
\section{}\label{ch1:appendix main proof}

This section contains the proofs of Theorem \ref{ch1:MCE}, Theorem \ref{ch1:SCE}, Proposition \ref{ch1:propsuff}, and some additional propositions and lemmas required for the proofs of the theorems.

First off, note that \ref{ch1:SP} implies \textit{Pareto indifference}.

\begin{axiomp}{PI}[Pareto Indifference]\label{ch1:PI}
For any $(S,R)\in \Omega$ and $x,y\in \triangle S$, if $x \sim^R_i y$ for all $i\in N$, then $x \sim^{(S,R)}_* y$.
\end{axiomp}

We use this axiom for our first interim result.

\begin{proposition}\label{ch1:propadditive}
Let \ref{ch1:RA} and \ref{ch1:PI} be satisfied. Then for any $(S,R)\in\Omega$ and for any representations $(u_i)_{i\in N}$ of $R$ and $u_*$ of $\succcurlyeq^{(S,R)}_*$, there exists $(\lambda_i)_{i=0}^{n} \in\mathbb{R}^{n+1}$, such that for all $a\in S$, $$u_*(a) = \lambda_0+ \sum_{i\in N} \lambda_i u_i(a).$$
\end{proposition}

\begin{proof}
Fix $(S,R)\in\Omega$. We denote the alternatives in $S$ by $a_1$ to $a_m$ where $m=|S|$. The pay-off vector $\vec{u}_i := (u_i(a_1),...,u_i(a_m))$ describes individual $i's$ vNM preferences over $S$. By \ref{ch1:RA}, the representation $u_*$ of $\succcurlyeq^{(S,R)}_*$ must be an expected utility representation and is therefore fully described by a pay-off vector as well. We denote this vector by $\vec{u}_*:=(u_*(a_1),...,u_*(a_m))$. Let \begin{align}\nonumber
    M &:= \begin{bmatrix}
    (1,...,1)\\
           \vec{u}_1 \\
           \vdots \\
           \vec{u}_n
         \end{bmatrix}
\end{align}
be the \textit{pay-off matrix}, which describes individual preferences over $S$. Note that the proposition states that $\vec{u}_*$ is equal to some linear combination of the rows in $M$. We distinguish two cases. In the first case, individual preferences are sufficiently diverse, such that there are $m$ linearly independent rows in $M$, which span the entire $m$-fold vector space. Hence, a linear combination of rows in $M$ equal to $\vec{u}_*$ exists trivially, even without invoking \ref{ch1:PI}. Next, consider the case where individual preferences are not sufficiently diverse, such that not every vector of length $m$ can be expressed as a linear combination of rows. In this case, the maximal number of linearly independent rows is strictly below $m$. Hence, the maximal number of linearly independent columns must be strictly below $m$ as well. As $M$ has $m$ columns, linearly dependent columns must exist. Sequentially drop linearly dependent columns from $M$ until all remaining columns are linearly independent. Denote the set of alternatives associated with the remaining columns by $B$. We call the alternatives in $B$ the \textit{independent alternatives} and the alternatives in $S\setminus B$ the \textit{dependent alternatives}. Note that in the resulting matrix, the rows span the entirety of the reduced vector space. Hence, by dropping the dependent alternatives, preferences are again sufficiently diverse to express any utility vector for the independent alternatives. This means that we can find a linear combination of rows of $M$ that matches $\vec{u}_*$ in the utilities for the independent alternatives. Now consider the dependent alternatives. For each $a\in S\setminus B$, there must exist a linear combination of columns $\gamma^a:B\to\mathbb{R}$ such that
\begin{equation}\label{ch1:lincomb}
\sum_{b\in B} \gamma^a(b)\vec{u}(b) = \vec{u}(a).
\end{equation} We will now show that each $\gamma^a$ can be decomposed into two lotteries $\gamma^a_+,\gamma^a_- \in \triangle S$ such that every individual is indifferent between these lotteries. Let $B^a_+:=\{b\in B:\gamma^a(b)\geq 0\}$ and $B^a_-:=\{b\in B:\gamma^a(b)< 0\}$. Since by definition the first row of $M$ consists only of $1$'s, $\sum_{b\in B} \gamma^a(b)=1$ and furthermore $k:=\sum_{b\in B^a_+} \gamma^a(b)=1-\sum_{b\in B^a_-} \gamma^a(b)$. Now define two lotteries 
\begin{align}\nonumber
\gamma^a_+(b) := \begin{cases}
  \frac{1}{k} \gamma^a(b)  & b\in B^a_+ \\
  0 & b\notin B^a_+
\end{cases} && \text{and} && \gamma^a_-(b) := \begin{cases}
  -\frac{1}{k} \gamma^a(b)  & b\in B^a_- \\
  \frac{1}{k} & b=a\\
    0 & b\notin B^a_- \cup \{a\}
\end{cases}
\end{align}
Note that for each $i\in N$, $\sum_{b\in S} \gamma^a_+(b) u_i(b) = \sum_{b\in S} \gamma^a_-(b) u_i(b)$, meaning each individual is indifferent between the two lotteries. By \ref{ch1:PI}, the planner must be indifferent as well. Therefore, $\sum_{b\in S} \gamma^a_+(b) u_*(b) = \sum_{b\in S} \gamma^a_-(b) u_*(b)$ and furthermore 
\begin{equation}\label{ch1:lincombplanner}
\sum_{b\in B} \gamma^a(b)u_*(b) = u_*(a).
\end{equation}
Compare this to Equation (\ref{ch1:lincomb}). The planner's utility for any dependent alternatives is determined from the independent alternatives in the same way as for every individual. Therefore, if a linear combination of rows of $M$ matches the planner's utilities for the independent alternatives, it must also match the utilities for the dependent alternatives.
\end{proof}

We say that a state $(S,R)\in \Omega$ is \textit{polar} if (i) $S$ has exactly $n$ elements, which we denote by $p_1$ to $p_n$, (ii) for every $i\in N$, $p_j \sim^R_i p_{k}$ for all $j,k\in N\setminus\left\{i\right\}$ and (iii) either $p_i \succ^R_i p_j$ for all $j\in N\setminus\left\{i\right\}$ or $p_i \sim^R_i p_{j}$ for all $j\in N\setminus\left\{i\right\}$. We say that $p_i$ is $i$'s \textit{polar alternative}.

\begin{proposition}\label{ch1:proppolar}
Let \ref{ch1:RA}, \ref{ch1:SP}, \ref{ch1:AN}, \ref{ch1:IIA} and \ref{ch1:MICA} be satisfied. Then for any polar state $(S,R)\in\Omega$, if $i,j\in N$ have a strict preference on $S$, then $p_i \sim_*^{(S,R)} p_j$.
\end{proposition}

\begin{proof}
Assume $(S,R)\in\Omega$ is polar and $i,j\in N$ have strict preferences on $S$. If $i=j$, the proposition is trivially satisfied, so assume $i\neq j$. We assume $p_i \succcurlyeq^{(S,R)}_* p_j$ and show that $p_j \succcurlyeq^{(S,R)}_* p_i$ is implied, which only leaves $p_i \sim^{(S,R)}_* p_j$. We prove the proposition by going through a sequence of preference profiles and menus. We use Roman numerals as subscripts to keep track of the different preference profiles and menus. Let $R_{\I}\in\mathcal{R}^n$ denote the permutation of $R$ where only preferences of $i$ and $j$ are permuted. By \ref{ch1:AN}, $p_i \succcurlyeq^{(S,R_{\I})}_* p_j$. Let $R_{\II}\in\mathcal{R}^n$ denote a preference profile, which agrees with $R_{\I}$ on $S$ and where there is $q_i,q_j \in S^c$ such that $q_i \sim^{R_{\II}}_k p_i$ and $q_j \sim^{R_{\II}}_k p_j$ for all $k\in N$. By our assumption that $A$ has at least $2n+4$ elements, such alternatives must exist. By \ref{ch1:IIA}, $p_i \succcurlyeq^{(S,R_{\II})}_* p_j$. Let $S_{\I}:=S \cup \{q_i,q_j\}$. Then $p_i \succcurlyeq^{(S_{\I},R_{\II})}_* p_j$ by \ref{ch1:MICA}. Furthermore, $p_i \sim^{(S_{\I},R_{\II})}_* q_i$ and $p_j \sim^{(S_{\I},R_{\II})}_* q_j$ by \ref{ch1:SP} and hence $q_i \succcurlyeq^{(S_{\I},R_{\II})}_* q_j$ by \ref{ch1:RA}. Let $S_{\II}:= S_{\I} \setminus \{p_i,p_j\}$. Then $q_i \succcurlyeq^{(S_{\II},R_{\II})}_* q_j$ by \ref{ch1:MICA}. Let $R_{\III}\in\mathcal{R}^n$ denote a preference profile which agrees with $R_{\II}$ on $S_{\II}$ and where $q_i \sim^{R_{\III}}_k p_j$ and $q_j \sim^{R_{\III}}_k p_i$ for all $k\in N$. Then $q_i \succcurlyeq^{(S_{\II},R_{\III})}_* q_j$ by \ref{ch1:IIA}, $q_i \succcurlyeq^{(S_{\I},R_{\III})}_* q_j$ by \ref{ch1:MICA} and $p_j \succcurlyeq^{(S_{\I},R_{\III})}_* p_i$ by \ref{ch1:SP} and \ref{ch1:RA}. Furthermore, $p_j \succcurlyeq^{(S,R_{\III})}_* p_i$ by \ref{ch1:MICA}. Note that $R_{\III}$ and $R$ agree on $S$. Therefore, $p_j \succcurlyeq^{(S,R)}_* p_i$ by \ref{ch1:IIA}, which implies $p_i \sim^{(S,R)}_* p_j$.\end{proof}

We now show two properties of binary relations that we will need for the proofs of the theorems.

\begin{lemma}\label{ch1:lemmaagree}
Let $\succcurlyeq$, $\succcurlyeq'$ and $\succcurlyeq''$ be binary relations over $\triangle S$. If $\succcurlyeq$ and $\succcurlyeq'$ agree on $B\subseteq S$ and $\succcurlyeq'$ and $\succcurlyeq''$ agree on $C\subseteq S$, then $\succcurlyeq$ and $\succcurlyeq''$ agree on $B \cap C$.
\end{lemma}

\begin{proof}
Consider any $B,C\subseteq S$ s.t. $B\cap C \neq \emptyset$ and any $x,y\in \triangle (B \cap C)$. If $\succcurlyeq$ and $\succcurlyeq'$ agree on $B$, then $x\succcurlyeq y$ if and only if $x\succcurlyeq' y$, and if $\succcurlyeq'$ and $\succcurlyeq''$ agree on $C$, then $x\succcurlyeq' y$ if and only if $x\succcurlyeq'' y$. Hence, for any $x,y\in \triangle (B \cap C)$, $x\succcurlyeq y$ if and only if $x\succcurlyeq'' y$, meaning $\succcurlyeq$ and $\succcurlyeq''$ agree on $B \cap C$.
\end{proof}

\begin{lemma}\label{ch1:lemmathree}
Let $\succcurlyeq$ be a binary relation over $\triangle S$ satisfying \ref{ch1:RA}. For any $u:S\to \mathbb{R}$, if for each distinct $b,c,d \in S$ the part of $\succcurlyeq$ on $\triangle \{b,c,d\}$ is represented by $u$, then $\succcurlyeq$ is represented by $u$.
\end{lemma}

\begin{proof}
Let $u$ satisfy the premise above. Select $b,c,d\in S$ such that $b\succ c$. If this is not possible, then $\succcurlyeq$ must be totally indifferent on $S$, in which case the proof is trivial. By assumption, $u$ represents $\succcurlyeq$ on $\triangle \{b,c,d\}$. Let $\hat u:S\to \mathbb{R}$ denote a representation of $\succcurlyeq$ where $\hat u(a)=u(a)$ for all $a\in\{b,c,d\}$. Now consider any $e\in S\setminus \{b,c,d\}$. By assumption, both $u$ and $\hat u$ represent $\succcurlyeq$ on $\triangle \{b,c,e\}$. Since vNM representations are unique up to a positive affine transformation, there must exist $\alpha \in \mathbb{R}_+$ and $\beta \in \mathbb{R}$ such that $u(a) = \alpha \hat u(a) + \beta$ for all $a\in\{b,c,e\}$. As $\hat u(b)=u(b)$ and $\hat u(c)=u(c)$, we find that $(1-\alpha)u(b)=(1-\alpha)u(c)$ which further implies $\alpha=1$ and $\beta=0$. Hence, $\hat u(e)=u(e)$. As this holds true for any $e\in S\setminus \{b,c,d\}$, we find that $\hat u = u$ and hence $u$ represents $\succcurlyeq$.
\end{proof}
\ \\
\noindent \textit{Proof of Theorem \ref{ch1:MCE}} \nopagebreak[4] \\
Let \ref{ch1:RA}, \ref{ch1:SP}, \ref{ch1:AN}, \ref{ch1:IIA} and \ref{ch1:MICA} be satisfied and assume that $|A|\geq 2n+4$. We will prove that for each $(S,R)\in \Omega$, $\succcurlyeq^{(S,R)}_*$ is represented by $$\sum_{i\in N} u_i^{S,R}.$$ 
Fix $(S,R) \in \Omega$ where $|S|\geq 3$. Later, we consider the case $|S|=2$. We will show that for any distinct $b,c,d\in S$, the part of $\succcurlyeq^{(S,R)}_*$ on $\triangle\{b,c,d\}$ is represented by $\sum_{i\in N} u^{S,R}_i$. It then follows from Lemma \ref{ch1:lemmathree} that $\succcurlyeq^{(S,R)}_*$ is represented by $\sum_{i\in N} u^{S,R}_i$. The proof is done in three steps. First, we define a sequence of states, connecting $(S,R)$ to a polar state. Second, we show that in the final state of that sequence, social preferences on $\triangle\{b,c,d\}$ are represented by $\sum_{i\in N} u^{S,R}_i$. Third, we show that social preferences in the final state agree with the social preferences in the initial state on $\{b,c,d\}$.

We begin by defining the aforementioned sequence of states. We use Roman numerals as subscripts to keep track of the different preference profiles and menus. If $S\neq A$, then define $S_{\I}:=S$ and denote an arbitrary alternative in $S^c$ by $e$. If, on the other hand, $S = A$, we construct $S_{\I}$ in the following way. As $|A|\geq 2n+4$, $S \setminus \{b,c,d\}$ must have at least $2n+1$ alternatives. Consequently, at least one alternative $e\in S \setminus \{b,c,d\}$ must be comparable relative to $S \setminus \{e\}$ under $R$. So let $S_{\I}:=S\setminus \{e\}$. Note that either way, there is an alternative $e$ in $S^c_{\I}$. Let $R_{\I} \in \mathcal{R}^n$ denote a preference profile that agrees with $R$ on $S_{\I}$ and where for each $i\in N$, $e \sim^{R_{\I}}_i a$ for all $a\in\{a\in S_{\I}:a\succcurlyeq^{R_{\I}}_i a'\text{ for all }a'\in S_{\I}\}$. Let $S_{\II} := S_{\I} \cup \{e\}$. Next, identify the smallest subset $S_{\III}\subseteq S_{\II}$ with the property that $\{b,c,d,e\}\subseteq S_{\III}$ and every $a\in S_{\II}\setminus S_{\III}$ is comparable relative to $S_{\III}$ under $R_{\I}$. Note that there are at most $n+4$ alternatives in $S_{\III}$, namely $b,c,d,e,$ and $n$ alternatives that are each worst for exactly one individual. Hence, as $|A|\geq 2n+4$, there are at least $n$ alternatives in $S^c_{\III}$. For any state $\omega\in \Omega$, we denote the set of individuals that are not totally indifferent on the menu by $N_\succ^{\omega}$. Let $R_{\II} \in \mathcal{R}^n$ denote a preference profile that agrees with $R_{\I}$ on $S_{\III}$ and where there exists $P \subseteq S_{\III}^c$ such that (i) $(P,R_{\II})$ is polar, where $p_i\in P$ denotes $i$'s polar alternative, (ii) for each $i\in N$, $i\in N_\succ^{(P, R_{\II})}$ if and only if $i\in N_\succ^{(S_{\III}, R_{\II})}$ and (iii) for each $i\in N$, $p_i \sim^{R_{\II}}_i a$ for all $a\in\{a\in S_{\III}:a\succcurlyeq^{R_{\II}}_i a'\text{ for all }a'\in S_{\III}\}$ and $p_j \sim^{R_{\II}}_i a$ for all $a\in\{a\in S_{\III}:a'\succcurlyeq^{R_{\II}}_i a\text{ for all }a'\in S_{\III}\}$ and $j\in N\setminus \{i\}$. This concludes the construction of the sequence of states.

Next, we show that the part of $\succcurlyeq_*^{(P\cup S_{\III},R_{\II})}$ on $\triangle\{b,c,d\}$ is represented by $\sum_{i\in N} u^{S,R}_i$. By Proposition \ref{ch1:propadditive}, there exists weights $(\lambda_i)_{i=0}^n\in\mathbb{R}^{n+1}$ such that $\succcurlyeq_*^{(P\cup S_{\III},R_{\II})}$ is represented by $$\lambda_0+ \sum_{i\in N} \lambda_i u^{P\cup S_{\III},R_{\II}}_i.$$ By Proposition \ref{ch1:proppolar}, $p_i \sim_*^{(P,R_{\II})} p_j$ for all $i,j\in N_\succ^{{(P,R_{\II})}}$. Because every alternative in $S_{\III}$ is comparable relative to $P$ under $R_{\II}$, by \ref{ch1:MICA} $p_i \sim_*^{(P\cup S_{\III},R_{\II})} p_j$ for all $i,j\in N_\succ^{{(P\cup S_{\III},R_{\II})}}$ as well. This implies $\lambda_i=\lambda_j=:\lambda$ for all $i,j\in N_\succ^{{(P\cup S_{\III},R_{\II})}}$. Note that if $i\notin N_\succ^{{(P\cup S_{\III},R_{\II})}}$, then $u_i^{P\cup S_{\III},R_{\II}}$ is constant on $P\cup S_{\III}$, and $\lambda_i$ can be normalized to $\lambda$ as the planner's vNM representation is unique up to positive affine transformations. Similarly, $\lambda_0$ can be normalized to 0. We now show that $\lambda>0$ and, therefore, $\lambda$ can be normalized to 1. We distinguish three cases. First, consider the case where $N_\succ^{{(P\cup S_{\III},R_{\II})}}=N$ and for any alternative $a\in S_{\III}$ there exists a lottery $x_a\in\triangle P$ such that $a \sim^{R_{\II}}_i x_a$ for all $i\in N$. Then $\succcurlyeq^{(P\cup S_{\III},R_{\II})}_*$ is totally indifferent on $P\cup S_{\III}$ and any $\lambda$ represents the same preferences. Second, consider the case where $N_\succ^{{(P\cup S_{\III},R_{\II})}}=N$ and for some alternative $a\in S$ such a lottery does not exist. Depending on whether $\sum_{i=1}^n u_i^{S_{\III},R_{\II}}(a)$ is greater or smaller than 1, one can construct a lottery over $P$ that either Pareto dominates $a$ or is Pareto dominated by $a$. Then for \ref{ch1:SP} to be satisfied, $\lambda$ has to be strictly positive. Third, if $N_\succ^{{(P\cup S_{\III},R_{\II})}}\neq N$ then for \ref{ch1:SP} to be satisfied, $p_i \succ^{{(P\cup S_{\III},R_{\II})}}_* p_j$ for any $i\in N_\succ^{{(P\cup S_{\III},R_{\II})}}$ and $j\notin N_\succ^{{(P\cup S_{\III},R_{\II})}}$. This requires $\lambda$ to be strictly positive as well. In any of the three cases, $\lambda$ can be normalized to 1. Hence, we have shown that $\succcurlyeq_*^{(P\cup S_{\III},R_{\II})}$ is represented by $$\sum_{i\in N} u^{P\cup S_{\III},R_{\II}}_i.$$ Finally, note that we have constructed $(P\cup S_{\III},R_{\II})$ in such a way that for each $i\in N$, $u_i^{P\cup S_{\III},R_{\II}}(a)=u_i^{S,R}(a)$ for all $a\in\{b,c,d\}$. Hence, we have shown that the part of $\succcurlyeq_*^{(P\cup S_{\III},R_{\II})}$ on $\triangle\{b,c,d\}$ is indeed represented by $\sum_{i\in N} u^{S,R}_i.$

In the third and final step, we show that $\succcurlyeq_*^{(P\cup S_{\III},R_{\II})}$ and $\succcurlyeq_*^{(S,R)}$ agree on $\{b,c,d\}$. \begin{align*} 
\succcurlyeq_*^{(S,R)}&\text{ and } \succcurlyeq_*^{(S_{\I},R)} \text{ agree on } \{b,c,d\}. \text{ (\ref{ch1:MICA})} \\ 
\succcurlyeq_*^{(S_{\I},R)} &\text{ and } \succcurlyeq_*^{(S_{\I},R_{\I})} \text{ agree on } \{b,c,d\}. \text{ (\ref{ch1:IIA})} \\ 
\succcurlyeq_*^{(S_{\I},R_{\I})} &\text{ and } \succcurlyeq_*^{(S_{\II},R_{\I})} \text{ agree on } \{b,c,d\}. \text{ (\ref{ch1:MICA})} \\ 
\succcurlyeq_*^{(S_{\II},R_{\I})} &\text{ and }\succcurlyeq_*^{(S_{\III},R_{\I})} \text{ agree on } \{b,c,d\}. \text{ (\ref{ch1:MICA})} \\ 
\succcurlyeq_*^{(S_{\III},R_{\I})}&\text{ and } \succcurlyeq_*^{(S_{\III},R_{\II})}\text{ agree on } \{b,c,d\}. \text{ (\ref{ch1:IIA})} \\ 
\succcurlyeq_*^{(S_{\III},R_{\II})}&\text{ and } \succcurlyeq_*^{(P\cup S_{\III},R_{\II})} \text{ agree on } \{b,c,d\}. \text{ (\ref{ch1:MICA})}
\end{align*} Hence, by Lemma \ref{ch1:lemmaagree} $\succcurlyeq_*^{(P\cup S_{\III},R_{\II})}$ and $\succcurlyeq_*^{(S,R)}$ agree on $\{b,c,d\}$. This concludes the proof for $|S|\geq 3$.

Now consider $|S|=2$. There must at least be $2n+2$ alternatives in $S^c$. We consider a different profile $R_{\I}$ where there is a set of polar alternatives $P$ outside of $S$. Then $\succcurlyeq_*^{(P\cup S,R_{\I})}$ must be represented by equal weights by the above argument. By \ref{ch1:MICA}, $\succcurlyeq_*^{(S,R_{\I})}$ must be represented by equal weights and by \ref{ch1:IIA} $\succcurlyeq_*^{(S,R)}$ must be represented by equal weights. This concludes the proof of Theorem \ref{ch1:MCE}.\\

Next, we derive an interim result required for the proof of Theorem \ref{ch1:SCE}.

\begin{proposition}\label{ch1:proppolar2}
Let \ref{ch1:RA}, \ref{ch1:SP}, \ref{ch1:AN}, \ref{ch1:MI}, and \ref{ch1:IICA} be satisfied. Then for any polar $(S,R)\in\Omega$, if every $a\in S^c$ is comparable relative to $S$ under $R$ and $i,j\in N$ have a strict preference on $S$, then $p_i \sim_*^{(S,R)} p_j$.
\end{proposition}

The proof is nearly identical to the proof of Proposition \ref{ch1:proppolar}, with the only difference that  \ref{ch1:IICA} is used instead of \ref{ch1:IIA}. This is possible as Proposition \ref{ch1:proppolar2} is restricted to polar states where all alternatives outside are comparable.\\
\ \\
\textit{Proof of Theorem \ref{ch1:SCE}} \nopagebreak[4] \\
Let \ref{ch1:RA}, \ref{ch1:SP}, \ref{ch1:AN}, \ref{ch1:IICA} and \ref{ch1:MI} be satisfied and assume that $|A|\geq 2n+4$. We show that for each $(S,R)\in\Omega$, $\succcurlyeq^{(S,R)}_*$ is represented by $$\sum_{i\in N} u_i^{A,R}.$$ Note that it suffices to show that this holds for $S=A$, as by \ref{ch1:MI} the same representation must hold for any $S\subset A$. Furthermore, note that for $S=A$, the representations of both theorems coincide. Hence, we follow the proof of Theorem \ref{ch1:MCE}, with the caveat that $R_{\II}$ is selected such that every alternative in $(P \cup S_{\III})^c$ is comparable relative to $(P \cup S_{\III})$ under $R_{\II}$. Then one can simply replace the use of Proposition \ref{ch1:proppolar} with Proposition \ref{ch1:proppolar2} and the use of \ref{ch1:IIA} with \ref{ch1:IICA} for the proof to go through. This concludes the proof of Theorem \ref{ch1:SCE}.\\
\ \\
\textit{Proof of Proposition \ref{ch1:propsuff}} \nopagebreak[4] \\
We prove the proposition by providing a counterexample. Specifically, we identify an aggregation rule that satisfies the axioms but is not represented by the normalized sum of individual utilities across all states. We begin with a counterexample for the representation of Theorem \ref{ch1:MCE}. Assume that \ref{ch1:RA}, \ref{ch1:SP}, \ref{ch1:AN}, \ref{ch1:IIA} and \ref{ch1:MICA} are satisfied. Fix a state $(A,\hat R)\in \Omega$ where no alternative is comparable relative to the remaining alternatives under $\hat R$. As $|A|<2n+1$, such a state must exist. Let $\pi(\hat R)$ denote the set containing all permutations of $\hat R$ and $\hat R$ itself and let $\hat \Omega:=\{(S,R)\in \Omega:S=A, R \in \pi(\hat R)\}$. Now consider an aggregation rule where $\succcurlyeq_*^{(S,R)}$ is represented by \begin{equation}\label{ch1:suffrep1}
\sum_{i\in N} \left(\sum_{a\in S} u^{S,R}_i(a)\right)u^{S,R}_i
\end{equation} whenever $(S,R)\in \hat \Omega$ and by \begin{equation}\label{ch1:suffrep2}\sum_{i\in N} u^{S,R}_i\end{equation} whenever $(S,R)\notin \hat \Omega$. Note that it could be that for all $(S,R)\in \hat \Omega$, (\ref{ch1:suffrep1}) and (\ref{ch1:suffrep2}) are positive affine transformations of each other. So assume that $(A,\hat R)$ has been selected such that this is not the case, which is possible by the richness of $\Omega$. If this aggregation rule indeed satisfies our axioms, we have produced a counterexample. Note that \ref{ch1:AN}, \ref{ch1:MICA} and \ref{ch1:IIA} \textit{connect} states, meaning they impose restrictions between states, whereas \ref{ch1:RA} and \ref{ch1:SP} impose restrictions on each state separately. So to prove that no axiom is violated, we will show that (i) no axiom connects a state in $\hat \Omega$ to a state outside of $\hat \Omega$, (ii) (\ref{ch1:suffrep1}) satisfies the restrictions imposed between any two states in $\hat \Omega$, and (iii) (\ref{ch1:suffrep1}) satisfies \ref{ch1:RA} and \ref{ch1:SP} for each $(S,R)\in\hat \Omega$. First, $\hat \Omega$ has been constructed such that \ref{ch1:AN} doesn't connect any state in $\hat \Omega$ to a state outside of $\hat \Omega$. \ref{ch1:IIA} doesn't connect any state in $\hat \Omega$ to another state, as there are no alternatives outside the menu. \ref{ch1:MICA} doesn't connect any state in $\hat \Omega$ to another state, as no alternative is comparable relative to the other alternatives in $A$ under any $R\in\pi(\hat R)$. Second, (\ref{ch1:suffrep1}) satisfies \ref{ch1:AN} as the weight on each utility function only depends on the utility function itself but not on the index. Third, (\ref{ch1:suffrep1}) assigns positive weights to all individual utility functions that are not indifferent on $A$. Hence, \ref{ch1:RA} and \ref{ch1:SP} are satisfied for each $(S,R)\in \hat \Omega$. This concludes the proof of Proposition \ref{ch1:propsuff} in case of Theorem \ref{ch1:MCE}.

Next, we provide a counterexample for the representation of Theorem \ref{ch1:SCE}. Let $\hat \Omega$ be defined as above. Now consider an aggregation rule where $\succcurlyeq_*^{(S,R)}$ is represented by \begin{equation}\label{ch1:suffrep3}
\sum_{i\in N} \left(\sum_{a\in S} u^{A,R}_i(a)\right)u^{A,R}_i
\end{equation} whenever $(S,R)\in \hat \Omega$ and by \begin{equation}\label{ch1:suffrep4}\sum_{i\in N} u^{A,R}_i\end{equation} whenever $(S,R)\notin \hat \Omega$. As before, assume that $(A,\hat R)$ has been selected such that (\ref{ch1:suffrep3}) and (\ref{ch1:suffrep4}) are not positive affine transformations of each other. Note that both \ref{ch1:MI} and \ref{ch1:IICA} connect states. As before, it holds that (i) no axiom connects a state in $\hat \Omega$ to a state outside of $\hat \Omega$, (ii) (\ref{ch1:suffrep3}) satisfies the restrictions imposed between any two states in $\hat \Omega$, and (iii) (\ref{ch1:suffrep3}) satisfies \ref{ch1:RA} and \ref{ch1:SP} for each $(S,R)\in\hat \Omega$. This can be shown, similarly to how it was shown for Theorem \ref{ch1:MCE} above. Just note that \ref{ch1:IICA} doesn't connect any state in $\hat \Omega$ to another state, as no alternative is comparable relative to the other alternatives in $A$ under any $R\in\pi(\hat R)$. This concludes the proof of Proposition \ref{ch1:propsuff} in the case of Theorem \ref{ch1:SCE}.

\section{}\label{ch1:appendix infinite alternatives}

In this section, we consider the case where $A$ is either countably or uncountable infinite. This requires us to make some adjustments to the framework and the axioms. First, individual utilities might not be bounded, in which case they cannot be normalized as in the representations of Theorems \ref{ch1:MCE} and \ref{ch1:SCE}. We deal with this by introducing a domain restriction, namely we only impose axioms on states where each individual has a best and worst alternative in $A$. Formally, we define $\Omega$ to be the set of states, such that for each $(S,R)\in\Omega$, both $\left\{a\in A:a\succcurlyeq^R_i b \text{ for all }b\in A\right\}$ and $\left\{a\in A:b\succcurlyeq^R_i a \text{ for all }b\in A\right\}$ are non-empty for all $i\in N$. Second, even if individual preferences have a best and worst alternative in $A$, they might not have one for $S\subset A$. For example, if $A=[0,1]$ and $u^R_i(a)=a$ for some $R\in \mathcal{R}^n$ and $i\in N$, then $\succcurlyeq^R_i$ does not have a best alternative in $S=[0,1)$. The non-existence of a best or worst alternative for an individual in a given menu means that we cannot construct the specific polar state required for the proof of Theorem \ref{ch1:MCE}. Therein, every alternative in the menu must be comparable relative to the set of polar alternatives and every polar alternative must be comparable relative to the menu. For the proof to go through, we introduce a weaker notation of comparability.

\begin{definition}
$a\in A$ is \textit{approximately comparable} relative to $B \subseteq A$ under $R\in\mathcal{R}^n$ if $a\notin B$ and for every $i\in N$ and $\varepsilon \in (0,1)$ there exists $x_{i,\varepsilon},y_{i,\varepsilon}\in \triangle B$ such that $$(1-\varepsilon)[a]+\varepsilon x_{i,\varepsilon} \sim^R_i y_{i,\varepsilon}.$$
\end{definition}

Note that if $a$ is comparable relative to $B$ under $R$, then $a$ is approximately comparable relative to $B$ under $R$. If $A$ is finite, the two concepts coincide. Furthermore, if $a$ is approximately comparable, its utility must lie between the supremum and infimum utility in the set for every individual, as shown by the following lemma.

\begin{lemma}\label{ch1:lemmaapprox}
For any $R\in\mathcal{R}^n$ and $B \subseteq A$, if $a\notin B$ and for any utility profile $(u^R_i)_{i\in N}$ of $R$, $$\sup_{b\in B} u_i^R(b) \geq u_i^R(a) \geq \inf_{b\in B} u_i^R(b)$$ for all $i\in N$ then $a$ is approximately comparable relative to $B$ under $R$.
\end{lemma}

\begin{proof}
Fix $(u^R_i)_{i\in N}$. Note that $a\in A$ is approximately comparable relative to $B \subseteq A$ under $R\in\mathcal{R}^n$ if and only if for every $i\in N$ and $\varepsilon \in (0,1)$ there exists $x_{i,\varepsilon},y_{i,\varepsilon}\in \triangle B$ such that \begin{equation}\label{ch1:convergenve} (1-\varepsilon)u_i^R(a) + \varepsilon u_i^R(x_{i,\varepsilon})=u_i^R(y_{i,\varepsilon}).\end{equation} If $u_i^R(a)$ is strictly between $\sup_{b\in B} u_i^R(b)$ and $\inf_{b\in B} u_i^R(b)$, one can simply select a $z_i \in \triangle B$ such that $u_i^R(z_i)=u_i^R(a)$ and then set $x_{i,\varepsilon}=y_{i,\varepsilon}=z_i$. Then (\ref{ch1:convergenve}) is satisfied for all $\varepsilon$. So assume $u_i^R(a)=\sup_{b\in B} u_i^R(b)$ and fix $\varepsilon$. Choose $x_{i,\varepsilon}$ arbitrarily. The left-hand side of (\ref{ch1:convergenve}) is strictly between $\sup_{b\in B} u_i^R(b)$ and $\inf_{b\in B} u_i^R(b)$, and hence there must be a $y_{i,\varepsilon}\in\triangle B$ to satisfy (\ref{ch1:convergenve}). The same argument applies when $u_i^R(a)=\inf_{b\in B} u_i^R(b)$.
\end{proof}

In \ref{ch1:MICA}, comparability is then replaced by approximate comparability. 

\begin{axiomp}{MICA*}\label{ch1:MICA*}
For each $(S,R)\in \Omega$ and $S'\subseteq S$ where every $a\in S\setminus S'$ is approximately comparable relative to $S'$, $\succcurlyeq^{(S,R)}_*$ and $\succcurlyeq^{(S',R)}_*$ agree on $S'$.
\end{axiomp}

With these adjustments, we can now state the equivalent of Theorem \ref{ch1:MCE} when $A$ is infinite. For any $R\in \mathcal{R}^n$ and $B \subseteq A$, let $\hat u^{B,R}_i$ denote the representation of $\succcurlyeq^R_i$ where $\sup_{a\in B} \hat u_i^{B,R}(a) = 1$ and $\inf_{a\in B} u_i^{B,R}(a) = 0$, unless $\succcurlyeq^R_i$ is indifferent on $B$ in which case $\hat u_i^{B,R}(a) = 0$ for all $a\in B$.

\begin{theorem}\label{ch1:MCE*}
Let $A$ be infinite. An aggregation rule $\succcurlyeq_*$ satisfies \ref{ch1:RA}, \ref{ch1:SP}, \ref{ch1:AN}, \ref{ch1:IIA} and \ref{ch1:MICA*} if and only if for each $(S,R)\in \Omega$, $\succcurlyeq^{(S,R)}_*$ is represented by $$\sum_{i\in N} \hat u_i^{S,R}.$$
\end{theorem}

\begin{proof}
First off, note that the proofs of Proposition \ref{ch1:proppolar}, Lemma \ref{ch1:lemmaagree} and Lemma \ref{ch1:lemmathree} go through when there are infinitely many possible alternatives. For a proof of Proposition \ref{ch1:propadditive} under infinite $A$ we refer to \textcite{mandler2005harsanyi}. Now consider the proof of Theorem \ref{ch1:MCE}. When constructing the sequence of states, specifically $R_{\I}$ and $R_{\II}$, there might not be a best or worst alternative in the menu for some individuals. We make the following adjustments to the proof. Let $R_{\I} \in \mathcal{R}^n$ denote a preference profile that agrees with $R$ on $S_{\I}$ and where there is an $e\in S^c_{\I}$ such that for each $i\in N$ and any $u_i^{R_{\I}}$ representing $\succcurlyeq_i^{R_{\I}}$, $u_i^{R_{\I}}(e) = \sup_{a\in S_{\I}} u_i^{R_{\I}}(a)$. Let $R_{\II} \in \mathcal{R}^n$ denote a preference profile that agrees with $R_{\I}$ on $S_{\III}$ and where there exists $P \subseteq S_{\III}^c$ such that (i) $(P,R_{\II})$ is polar, where $p_i\in P$ denotes $i$'s polar alternative, (ii) for each $i\in N$, $i\in N_\succ^{(P, R_{\II})}$ if and only if $i\in N_\succ^{(S_{\III}, R_{\II})}$ and (iii) for each $i\in N$ and any $u_i^{R_{\II}}$ representing $\succcurlyeq_i^{R_{\II}}$, $u_i^{R_{\II}}(p_i) = \sup_{a\in S_{\III}} u_i^{R_{\II}}(a)$ and $u_i^{R_{\II}}(p_j) = \inf_{a\in S_{\III}} u_i^{R_{\II}}(a)$ for all $j\in N \setminus \{i\}$. By Lemma \ref{ch1:lemmaapprox}, $e$ is approximately comparable to $S_{\I}$ under $R_{\I}$ and every alternative in $P$ is approximately comparable relative to $S_{\III}$ under $R_{\II}$. Then the proof of Theorem \ref{ch1:MCE} goes through. 
\end{proof}

Theorem \ref{ch1:SCE} holds without adjustment to the axioms.

\section{}\label{ch1:appendix droping axioms}

In this section, we show that none of the axioms can be dispensed with. We do so by separately dropping each axiom and then identifying a representation that satisfies the remaining axioms, different from the representation of the respective theorem.

When \ref{ch1:SP} is dropped in either theorem, total indifference of the planner satisfies the remaining axioms. Formally, for each $(S,R)\in\Omega$ and $x\in\triangle S$, $u_*^{(S,R)}(x)=0$. When \ref{ch1:AN} is dropped in Theorem \ref{ch1:MCE}, any weighted sum of individual utilities satisfies the remaining axioms as long as the weights are strictly positive. Formally, there exists $(\lambda_i)_{i\in N} \in \mathbb{R}_+^n$ such that for each $(S,R)\in \Omega$, $u_*^{(S,R)}(x)= \sum_{i=1}^n \lambda_i u^{S,R}_i(x)$. Similarly, when \ref{ch1:AN} is dropped in Theorem \ref{ch1:SCE}, $u_*^{(S,R)}(x)= \sum_{i=1}^n \lambda_i u^{A,R}_i(x)$ satisfies the remaining axioms. If \ref{ch1:RA} is dropped in either theorem, then the relative egalitarian rule by \textcite{sprumont2013relative} satisfies the remaining axioms. The rule ranks alternatives according to the lowest normalized individual utility it generates. If two alternatives generate the same lowest utility, they are further ranked by the second lowest utility, and so on. Depending on the theorem for which we want to provide the counterexample, individual utilities are normalized either with respect to $S$ or $A$. When \ref{ch1:IIA} is dropped in Theorem \ref{ch1:MCE}, then the representation of Theorem \ref{ch1:SCE} satisfies the remaining axioms. Similarly, when \ref{ch1:MI} is dropped in Theorem \ref{ch1:SCE}, the then representation of Theorem \ref{ch1:MCE} satisfies the remaining axioms. For a representation when either \ref{ch1:MICA} is dropped in Theorem \ref{ch1:MCE} or \ref{ch1:IICA} is dropped in Theorem \ref{ch1:SCE} we refer to Propositions \ref{ch1:prop minimal suff sce} and \ref{ch1:prop minimal suff mce} in Appendix \ref{ch1:appendix minimal sufficiency}.

\section{}\label{ch1:appendix minimal sufficiency}

In this section we consider weaker versions of our context-defining axioms \ref{ch1:MICA} and \ref{ch1:IICA}. The following notion of redundant alternatives has been considered by the literature.

\begin{definition} $a\in A$ is \textit{redundant} relative to $B \subseteq A$ under $R\in\mathcal{R}^n$ if $a\notin B$ and there exists $x\in\triangle B$ such that $x \sim^R_i [a]$ for all $i\in N$.
\end{definition}

Given this definition, we can state the central axiom of \textcite{dhillon1999relative}. Note that we adapted the axiom to our setting where the menu is explicit.

\begin{axiomp}{IRA}[Independence of Redundant Alternatives]\label{ch1:IRA}
Fix $(S,R)\in \Omega$ such that every $a\in S^c$ is redundant relative to $S$ under $R$. For any $R'\in \mathcal{R}^n$, if $R$ and $R'$ agree on $S$ and every $a\in S^c$ is redundant relative to $S$ under $R'$, then $\succcurlyeq_*^{(S,R')} = \succcurlyeq_*^{(S,R)}$.
\end{axiomp}

Note that \ref{ch1:IRA} differs from \ref{ch1:IICA} in two ways. First, the independence applies to redundant rather than comparable alternatives. Second, \ref{ch1:IRA} applies only if all alternatives outside the menu are redundant, whereas \ref{ch1:IICA} applies when only some alternatives outside the menu are comparable as well. In order to disentangle these difference, we define two intermediate axioms.

\begin{axiomp}{IICA*}\label{ch1:IICA*}
Fix $(S,R)\in \Omega$ such that every $a\in S^c$ is comparable relative to $S$ under $R$. For any $R'\in \mathcal{R}^n$, if $R$ and $R'$ agree on $S$ and every $a\in S^c$ is comparable relative to $S$ under $R'$, then $\succcurlyeq_*^{(S,R')} = \succcurlyeq_*^{(S,R)}$.
\end{axiomp}

\begin{axiomp}{IIRA}[Independence of Irrelevant Redundant Alternatives]\label{ch1:IIRA}
Fix $(S,R)\in \Omega$ and $C\subseteq S^{c}$ such that every $a\in C$ is redundant relative to $C^c$ under $R$. For any $R'\in \mathcal{R}^n$, if $R$ and $R'$ agree on $C^c$ and every $a\in C$ is redundant relative to $C^c$ under $R'$, then $\succcurlyeq_*^{(S,R')} = \succcurlyeq_*^{(S,R)}$.
\end{axiomp}

It is easy to see that \ref{ch1:IICA*} would suffice in Theorem \ref{ch1:SCE}, since whenever \ref{ch1:IICA} is used in the proof, all alternatives outside the menu are comparable. However, replacing \ref{ch1:IICA} with \ref{ch1:IIRA} would not suffice, as demonstrated by the following proposition. Consequently, replacing \ref{ch1:IICA} with \ref{ch1:IRA} wouldn't suffice either, as \ref{ch1:IRA} is weaker than \ref{ch1:IIRA}.

\begin{proposition}\label{ch1:prop minimal suff sce} Let $|A| \geq 2n +4$. There exists an aggregation rule $\succcurlyeq_*$ that satisfies \ref{ch1:RA}, \ref{ch1:SP}, \ref{ch1:AN}, \ref{ch1:IIRA} and \ref{ch1:MI} and $\succcurlyeq_*^{(S,R)}$ can not be represented by $\sum_{i\in N} u^{A,R}_i$ for all $(S,R)\in\Omega$.
\end{proposition}

\begin{proof}
The following proof is adapted from \textcite{sprumont2013relative}. We prove the proposition by identifying a representation that satisfies the axioms but is not a positive affine transformation of the normalized sum of individual utility functions. For every $(S,R)\in \Omega$, define $$\Psi(S,R):= \argmax_{x\in \triangle S} \prod_{i\in N} (u^{S,R}_i(x)+1)$$ and note that $\Psi(S,R)$ is non-empty and for every $x,y\in\Psi(S,R)$, $u^{S,R}_i(x)=u^{S,R}_i(y)$ for all $i\in N$. Hence, $\psi_i^{(S,R)}:=u^{S,R}_i(x)+1$ for some $x\in \Psi(S,R)$ is well defined. Now consider the following representation of $\succcurlyeq_*^{(S,R)}$, $$\sum_{i\in N} \psi^{(A,R)}_i u^{A,R}_i.$$ As $\psi^{(A,R)}_i$ is not necessarily the same for all $i\in N$, the representation is not necessarily a positive affine transformation of $\sum_{i\in N} u^{A,R}_i$. What remains to show is that the representation satisfies the axioms stated in the proposition. \ref{ch1:RA} and \ref{ch1:SP} are satisfied as the representation is a weighted sum of individual utility functions with strictly positive weights. \ref{ch1:AN} is satisfied as any individual's weight does not depend on the individual's index. \ref{ch1:MI} is satisfied as the weights do not depend on the menu. Finally, to see that \ref{ch1:IIRA} is satisfied, assume that for $(S,R)\in \Omega$ there exists a set of  redundant alternatives $C\subseteq S^c$. Now consider $R'\in \mathcal{R}^n$ where $R'\neq R$, $R$ and $R'$ agree on $C^c$ and $C$ is still redundant under $R'$. As any redundant alternative has a lottery that yields the same product, redundant alternatives can be ignored in the maximization problem involved in $\Psi(S,R)$. Hence, $(\psi_i^{(A,R)})_{i\in N}=(\psi_i^{(A,R')})_{i\in N}$. Note that \ref{ch1:IICA} is violated as changing individual preferences over comparable alternatives outside the menu can affect the maximization problem and, consequently, the weights.
\end{proof}

Similarly, we can replace the notion of comparable alternatives in \ref{ch1:MICA} with redundant alternatives, resulting in the following axiom.

\begin{axiomp}{MIRA}[Menu Independence of Redundant Alternatives]\label{ch1:MIRA}
For each $(S,R)\in \Omega$ and $S'\subseteq S$ where every $a\in S\setminus S'$ is redundant relative to $S'$, $\succcurlyeq^{(S,R)}_*$ and $\succcurlyeq^{(S',R)}_*$ agree on $S'$.
\end{axiomp}

Analogously, \ref{ch1:MIRA} does not suffice in Theorem \ref{ch1:MCE}.

\begin{proposition} \label{ch1:prop minimal suff mce} Let $|A| \geq 2n +4$. There exists an aggregation rule $\succcurlyeq_*$ that satisfies \ref{ch1:RA}, \ref{ch1:SP}, \ref{ch1:AN}, \ref{ch1:IIA} and \ref{ch1:MIRA} and $\succcurlyeq_*^{(S,R)}$ can not be represented by $\sum_{i\in N} u^{S,R}_i$ for all $(S,R)\in\Omega$.
\end{proposition}

A representation of $\succcurlyeq_*^{(S,R)}$ that satisfies the stated axioms is $\sum_{i\in N} \psi^{(S,R)}_i u^{S,R}_i.$ The proof is analogous to that of Proposition \ref{ch1:prop minimal suff sce}.
\end{appendices}

\printbibliography[segment=\therefsegment, heading=subbibintoc]

\end{document}